\documentclass[journal]{IEEEtran}

\usepackage{mathptmx}       % selects Times Roman as basic font
\usepackage{helvet}         % selects Helvetica as sans-serif font
\usepackage{courier}        % selects Courier as typewriter font
\usepackage{type1cm}        % activate if the above 3 fonts are
\usepackage{makeidx}         % allows index generation
\usepackage{graphicx}        % standard LaTeX graphics tool

\usepackage{multicol}        % used for the two-column index
\usepackage[bottom]{footmisc}% places footnotes at page bottom
\usepackage{ifthen}
\usepackage{subfigure}
\usepackage[usenames,dvipsnames]{color}
\usepackage[noadjust]{cite}
\makeindex             % used for the subject index
\usepackage{graphics} % for pdf, bitmapped graphics files
\usepackage{graphicx} % for pdf, bitmapped graphics files
\usepackage{epsfig} % for postscript graphics files
\usepackage{times} % assumes new font selection scheme installed
\usepackage{mathtools}  % loads amsmath
\usepackage{amssymb}  % assumes amsmath package installed
\usepackage{amsthm}

\usepackage{amsfonts}

\DeclareGraphicsExtensions{.pdf,.png,.jpg,.eps}

\usepackage{subfig}
\usepackage{verbatim}
\usepackage{comment}
\usepackage{algorithm}
\usepackage{algorithmic}
\usepackage{multirow}

\usepackage{threeparttable}
\usepackage[table]{xcolor}
\usepackage{colortbl}
\newcommand{\rr}{{\mathbb{R}}}
\newcommand{\mW}{{\mathcal{W}}}

\renewcommand{\algorithmicrequire}{{\bf Input:}}
\renewcommand{\algorithmicensure}{{\bf Output:}}

\theoremstyle{definition}
\newtheorem{theorem}{Theorem}
\newtheorem{lemma}[theorem]{Lemma}
\newtheorem{corollary}[theorem]{Corollary}

\theoremstyle{definition}

\newtheorem{example}{Example}
\newtheorem{assumption}{Assumption}
\newtheorem{property}[theorem]{Property}
\allowdisplaybreaks

\begin{document}

	\title{Approaching the Transient Stability Boundary of a Power System: Theory and Applications}
	
	\author{Peng Yang, 
		Feng Liu,~\IEEEmembership{Senior Member,~IEEE,}
		Wei Wei,~\IEEEmembership{Senior Member,~IEEE,} and
		Zhaojian Wang
		%	,~\IEEEmembership{ Member,~IEEE}
		%-	Steven~H.~Low,~\IEEEmembership{Fellow,~IEEE,}
		%-	Shengwei~Mei,~\IEEEmembership{Fellow,~IEEE,}		%Krishnamurthy~Dvijotham% <-this % stops a space
		%	\author{XX XXX,~\IEEEmembership{Member,~IEEE,}
		%		\thanks{This work was supported  by the National Natural Science Foundation
		%			of China ( No. 51321005, No. 51377092), Foundation of Chinese	Scholarship Council (CSC No. 201506215034) Los Alamos National Lab through an DoE grant DE-AC52-06NA25396, and Skoltech through Collaboration
		%			Agreement 1075-MRA.	}	% <-this % stops a space
		\thanks{P. Yang, F. Liu, and W. Wei are with the Department of Electrical Engineering, Tsinghua University, Beijing, China, 100084. Z. Wang is with the Department of Automation, Shanghai Jiao Tong University, Shanghai, China, 200240. Corresponding author: F. Liu (lfeng@tsinghua.edu.cn).}% <-this % stops a space
	}
	
	\maketitle
	
	\begin{abstract}
		Estimating the stability boundary is a fundamental and challenging problem in transient stability studies. It is known that a proper level set of a Lyapunov function or an energy function can provide an inner approximation of the stability boundary, and the estimation can be expanded by trajectory reversing methods. In this paper, we streamline the theoretical foundation of the expansion methodology, and generalize it by relaxing the request that the initial guess should be a subset of the stability region. We investigate topological characteristics of the expanded boundary, showing how an initial guess can approach the exact stability boundary locally or globally. We apply the theory to transient stability assessment, and propose expansion algorithms to improve the well-known Potential Energy Boundary Surface (PEBS) and Boundary of stability region based Controlling Unstable equilibrium point (BCU) methods. Case studies on the IEEE 39-bus system well verify our results and demonstrate that estimations of the stability boundary and the critical clearing time can be significantly improved with modest computational cost.   
	\end{abstract}
	

	\IEEEpeerreviewmaketitle

	\section{ INTRODUCTION}
	\label{sec:1}
	%	\subsection{ Background}
	\IEEEPARstart{T}{ransient} stability analysis is of crucial importance for power systems' security. One of the most momentous issues in this field is estimating the stability region and its boundary for a post-fault equilibrium \cite{Kundur_Definitionclassificationpower_2004}. The challenge is exacerbated by the transition of power systems to highly complex nonlinear systems dominated by massive renewable and distributed energies \cite{7419922}. In this paper, we aim to alleviate this issue by proposing an expansion methodology that improves the estimations obtained from the direct methods.

	Roughly speaking, the direct methods estimate the stability boundary by a level set of a Lyapunov function or an energy function \cite{ref:Chiang:2010,8444083,huang2021neural,8791570}. Compared with other methods that relies on solving system trajectories \cite{Stott_Powersystemdynamic_1979}, the direct methods avoid computationally costly numerical integration, and hence are much faster and enable on-line applications \cite{ref:Gless:1966}. However, estimations resulting from this approach often suffer from being too conservative and hence harm economic benefits.
	
	%	In the direct methods, one key element that determines the estimation accuracy is the Lyapunov (energy) function. Constructing a Lyapunov (energy) function for a nonlinear system is generally difficult. Such functions are traditionally related to the physical energy of the dynamic system \cite{1099743, Moon_Developmentenergyfunction_1997}. Recently, the sum-of-square and semi-definite programming (SDP) techniques have been devolved to construct Lyapunov function for polynomial systems and to estimate the stability region \cite{Chesi_DomainAttractionAnalysis_2011,Topcu_RobustRegionofAttractionEstimation_2010}. However, the size of the corresponding SDP grows quickly with the system order, so is difficult to cope with large-scale power systems.  Sampling-based methods are also proposed to construct Lyapunov function \cite{7798585,7810344} and to estimate the stability region \cite{Bobiti_AutomatedSamplingBasedStabilityVerification_2018}, which are time-consuming and only applicable to low-dimensional systems so far.
	
	One key element that determines the estimation accuracy is the level value of the estimation level set, i.e., the critical level value.
	Depending on the selection of the critical level value, the direct methods can be roughly classified into two categories: the global methods and the local methods. %The former have been investigated for transient stability analysis in power systems for decades \cite{ref:Kaye:1982, ref:Wu:1983, ref:Pavella:1985, ref:IEEE:1988}. 
	The former involves a level set contained completely in the stability region and estimates the stability boundary in all directions, i.e., a global estimation. %Although they are capable of providing global approximation, such methods are quite conservative. 
	The largest possible level value in this kind is given by the value at the closest unstable equilibrium point (UEP) as shown by Chiang et al. \cite{ref:Chiang:1989}.  Nevertheless, since the closest UEP has minimal value on the entire stability boundary, such a global method is still too conservative for industrial practice. 
	
	The local methods only estimate a certain part of the stability boundary relevant to a specific fault. In this category, a larger critical level value is used than the global estimation. Therefore, only part of the level set is contained in the stability region that serves as a local estimation. Compared with the global methods, the local methods greatly reduce conservativeness when only local information of the stability boundary is needed. 
	
	Among all local methods, the PEBS method and the BCU method stood out. 
	The PEBS method was first proposed by Kakimoto et al. \cite{ref:PEBS1,ref:PEBS2}, and was later endowed with a solid theoretic foundation \cite{Chiang_Foundationspotentialenergy_1988}. The basic idea of the PEBS method is to regard the first local maximum of the potential energy along the fault-on trajectory as the critical level value. It avoids calculations of UEP and reduces the system to a gradient one, contributing to a fast and simple algorithm for practical implementation. Nevertheless, the result of PEBS method could be over-optimistic and the correctness is not guaranteed due to the model reduction \cite{Chiang_Foundationspotentialenergy_1988}.
	The BCU method was established by Chiang et al. \cite{Chiang_BCUmethoddirect_1994,Chiang_TheoreticalfoundationBCU_1995}, using the concept of controlling UEP (CUEP) to get a fault-depended critical level value. Due to its solid theoretic foundation and satisfactory performance in practice, the BCU method is widely recognized as the most effective direct method \cite{Alberto_Directmethodstransient_2001}. %However, as this method is only a local approximation in the neighborhood of CUEP, the accuracy is harmed if the exiting point of the fault-on trajectory is distant from the CUEP. 
	%\subsection{ Motivations and Contributions}
	
	With a fixed critical level value, the monotonic property of the trajectories has long been used to reduce the conservativeness, which to the best of our knowledge dates back to the early 1980's \cite{ref:Genesio:1985}. In \cite{ref:Genesio:1985}, Genesio et al. suggested a \textit{trajectory reversing method} (TRM), indicating that the estimated region can be constantly enlarged by reverse integration and eventually approaches the exact stability region. However, the TRM algorithm relies on numerical integration starting from numerous points around the stable equilibrium point. Thus, no analytical results can be obtained and it is extremely computationally intensive, inhibiting applications to high-dimension systems.
	%	In \cite{ref:Jin:2005} and \cite{ref:Jin:2010}, the backward reachable set of a SEP is calculated by solving a set of Hamilton-Jacobi-Issac partial differential equations numerically. Most recently, a decomposition technique was proposed to alleviate the computational complexity \cite{8267187} and the forward reachable sets are used to estimate the DOA \cite{ref:reachable_set}. The reachability analysis is also studied for index-1 differential-algebraic systems \cite{Althoff_ReachabilityAnalysisNonlinear_2014}. These reachability-analysis-based methods are theoretically equivalent to the TRM because the reachable set is essentially determined by and only by the integral trajectories. It, therefore, provides accurate estimation, however, is time-consuming. It is reported in \cite{ref:reachable_set} that it takes 382.71s to calculate the 2-dimensional DOA of a simple single-machine power system.
 Chiang et al. first extended the idea of TRM and proposed a symbolic-calculation-based constructive method to improve the estimation of stability boundary in \cite{Chiang_constructivemethoddirect_1988}, and further developed this inspiring idea in \cite{ref:Chiang:1989-2} and \cite{ref:Chiang:1990}.   
	Our previous work \cite{ref:Liu:2011,Liu_estimationstabilityboundaries_2011} interpreted the expansion methodology from a geometrical view, and allowed more general iterative algorithms. Similar expansion ideas have also been used to produce closed-form stability region estimations for fuzzy-polynomial systems \cite{Pitarch_ClosedFormEstimatesDomain_2014} and to compute Lyapunov functions \cite{Doban_ComputationLyapunovFunctions_2018}. 
	
	Despite fruitful results of descending from the TRM methodology, several critical issues remain unaddressed and restrict its application in power systems. First, it requires a strict subset of the stability region to act as the initial guess. However, as the exact stability region is unknown beforehand, it is hard to certify this condition \textit{a prior}. Second, it can only improve the global estimation of the stability boundary, and hence cannot be applied to the more effective and wildly-used local methods, e.g., the PEBS and BCU methods. These issues motivate us to generalize the theoretic fundamentals of the expansion method and develop efficient algorithms that enable power system applications.
	
	Inspired by \cite{ref:Genesio:1985, Chiang_constructivemethoddirect_1988,ref:Chiang:1989-2,ref:Chiang:1990,ref:Liu:2011,Liu_estimationstabilityboundaries_2011}, we further develop the expansion methodology with special focus on power system applications. The contributions of this paper are mainly two folds: 
	\begin{enumerate}
		\item We consolidate theories for the expansion methodology, both locally and globally. By revealing the topological characteristics of the expansion, we relax the restrictive assumption in the current literature that requires the initial level set should be strictly contained in the exact stability region \cite{ref:Genesio:1985, Chiang_constructivemethoddirect_1988,ref:Chiang:1989-2,ref:Chiang:1990,ref:Liu:2011,Liu_estimationstabilityboundaries_2011}. Hence, the expansion methodology can be applied to local estimations, which are more wildly-used in power system applications.
		\item  We propose effective algorithms to improve estimations of the stability boundary and the critical clearing time (CCT) for power system transient stability analysis, which extend pioneering schemes in \cite{Chiang_constructivemethoddirect_1988,ref:Chiang:1989-2,ref:Chiang:1990} to a general expansion scheme. Such algorithms can significantly improve the widely acknowledged PEBS and BCU methods with modest extra computational effort, enabling applications to complex power systems.
	\end{enumerate} 
	
	The remainder of this paper is organized as follows. Section II presents necessary notations and preliminaries; Section III shows our main theoretical results with several illustrative examples; Section IV designs the expansion algorithms; Section V reports the case studies on the IEEE 39-bus power system benchmark; Finally, Section VI concludes the paper. Most proofs are deferred to the Appendix.

	\section{ Notation and Preliminary }
	\label{sec:2}
	Consider an autonomous continuous-time nonlinear system
	\begin{equation}
	\label{eq:system.1} 
	\dot{x}=f(x) 
	\end{equation}
	where  $f:\rr^n\to\rr^n$ is the vector field and $x(t)\in \rr^n$ is the state. We assume $f$ is sufficiently smooth such that solutions are complete for all initial conditions.  
	
	The solution of \eqref{eq:system.1} starting from $x_0 \in \rr^n$, i.e., $x(0)=x_0$, are called the flow of \eqref{eq:system.1} with respect to $x_0$, which is denoted by $\phi _{t}^{f} (x_0)$. 
	For any given $t\in \rr$, $\phi_t^f (\cdot)$  defines a mapping $\rr^n\to\rr^n$, which is referred to as the \textit{flow mapping}. 
	%	\begin{property} \label{Prop:1}
	%		\cite{Kelley_TheoryDifferentialEquations_2010} The flow mapping of system \eqref{eq:system.1}, $\phi^f_t: \rr^n\to\rr^n$, has the following properties:
	%		\begin{enumerate}
	%			\item For any given $t\in \rr $, it is a diffeomorphic mapping. 	
	%			\item It defines an additive group with an ``addition" operation ``$\circ$'' on $\rr$, i.e, for any $x\in\rr^n$	and any $t, \tau\in \rr$, there are
	%			\begin{equation*}
	%			\begin{aligned}
	%			&\phi^f_t \circ \phi^f_{\tau} \circ x =\phi^f_{t+\tau} \circ x=\phi^f_{\tau} \circ \phi^f_t \circ x\\
	%			&\phi^f_t \circ   \phi^f_{-t} \circ x =\phi^f_0 \circ x=x		
	%			\end{aligned}
	%			\end{equation*}
	%		\end{enumerate}
	%	\end{property}
	
	A point $x_{e}\in\rr^n $ is called an equilibrium point (EP) of \eqref{eq:system.1} if $f(x_{e})= 0$. Let $E:=\left\{x\in \rr^n|\; f\left(x\right) = 0\right\}$ denote the set of all EPs. We say $x_{e}\in E$ is a hyperbolic equilibrium point if the Jacobian of \eqref{eq:system.1} at $x_e$, $J=\frac{\partial f}{\partial x}|_{x=x_{e} } $, has no zero-real-part eigenvalue.
	For a hyperbolic EP, it is an asymptotically Lyapunov stable equilibrium point, denoted by ASEP, if all the eigenvalues of the Jacobian at this EP have negative real parts. Otherwise, it is unstable and is denoted by UEP. If the Jacobian at $x_{e} $ has exactly $k$ eigenvalues with positive real parts, $x_{e}$ is called a type-$k$ UEP.
	Its stable and unstable manifolds $\mW^{s} \left(x_{e} \right)$, $\mW^{u} \left(x_{e} \right)$ are defined as
	\begin{eqnarray*}
		\mW^{s} (x_{e} ) &:=& \{ x \in \rr^n| \mathop {\lim }\limits_{t\to +\infty } \phi _{t}^{f} (x)=x_{e} \} \\
		\mW^{u} (x_{e} ) &:=& \{ x \in \rr^n|\mathop{\lim }\limits_{t\to -\infty } \phi _{t}^{f} (x)=x_{e} \}.
	\end{eqnarray*}
	%	Both the stable and unstable manifolds are closed \cite{Wiggins_IntroductionAppliedNonlinear_2003}.
	
	%	\begin{definition}
	Given an ASEP $x_s$ of \eqref{eq:system.1}, its stability region (or region of attraction), denoted by $A(x_s)$, is
	\begin{equation}\label{eq:As}
	A(x_{s} ):=\{ x \in  \rr^n|\mathop{\lim }\limits_{t\to +\infty } \phi _{t}^{f} (x)=x_{s} \}. 
	\end{equation}
	%	\end{definition}
	$A(x_s)$ is a connected open invariant set and is homeomorphic to $\rr^n$ \cite{ref:Zaborszky:1988}. We use $\partial A(x_s)$ to denote its boundary. For simplicity, we assume there exists no closed orbit of \eqref{eq:system.1} on $\partial A(x_s)$, which is commonly satisfied in power systems with energy functions \cite{ref:Chiang:1989}.
	
	We follow the line of \cite{ref:Chiang:1988,ref:Chiang:survey1995,ref:Chiang:2010} and assume that the system \eqref{eq:system.1} satisfies the following assumptions.

	\begin{assumption}
		\label{as:1}
		\cite{ref:Chiang:1988} The equilibrium of system \eqref{eq:system.1} satisfy: 
		\begin{enumerate}
			\item All EPs on the stability boundary are hyperbolic;
			\item  The stable and unstable manifolds of EPs on stability boundaries satisfy the transversality condition
			%			\footnote{Consider two manifolds $\mathcal{X}_1,\mathcal{X}_2\subset\rr^n$, we say $\mathcal{X}_1$ and $\mathcal{X}_2$ satisfy the transversality condition if (i) $\forall x\in\mathcal{X}_1\cap\mathcal{X}_2$, the tangent spaces of $\mathcal{X}_1$ and $\mathcal{X}_2$ at $x$ span $\rr^n$, i.e. $T_x(\mathcal{X}_1)+T_x(\mathcal{X}_2)=\rr^n$, or (ii) $\mathcal{X}_1\cap\mathcal{X}_2=\emptyset$ };
			\item Every trajectory on  stability boundaries approaches one EP as $t\to \infty$. 
		\end{enumerate} 
	\end{assumption}

	For systems satisfying Assumption \ref{as:1}, its stability boundary is composed of stable manifolds of all UEPs on the boundary, as stated in the following theorem.
	
	\begin{theorem} \label{thm:boundary.1}
		\cite{ref:Chiang:1988} (\textit{Characteristics of the stability boundary}): If Assumption \ref{as:1} is satisfied, then
		\begin{eqnarray}
		\label{eq:boundary.2}
		\partial A(x_{s} ) = \bigcup _{x_u \in E\bigcap \partial A(x_{s} )} \mW^{s} (x_u ) 
		\end{eqnarray}
	\end{theorem}
	
	We refer interested readers to \cite{fisher2020hausdorff,fisher2021comments} for a more recent and complete discuss on conditions for Theorem \ref{thm:boundary.1}.
	
	Although \eqref{eq:boundary.2} characterizes the topological property of the stability boundary, the analytical expression for $\partial A(x_s)$ is unavailable for general power systems. Alternatively, it is able to estimate $\partial A(x_s)$ by proper level sets of Lyapunov functions or energy functions of the system \cite{ref:Chiang:survey1995,Khalil_NonlinearSystems_2002}.
	
	\begin{assumption}\label{as:2}
		There exists a function $V:\rr^n\to\rr$ that is a Lyapunov function, defined as in \cite[Theorem 4.1]{Khalil_NonlinearSystems_2002}, or is an energy function, defined as in \cite{ref:Chiang:survey1995}.
	\end{assumption}
	
	Given some positive number $l>0$, define
	\begin{equation}
	\label{eq:Sl.1}
	S_{l}:=\left\lbrace  x\in \rr^{n}|\;V(x)<l\right\rbrace.  
	\end{equation}
	Then $S_{l}$ is an open set and is referred to as the \textit{sublevel set} with \textit{level value} $l$. Let $\partial S_{l}$ denote the boundary of $S_{l}$, i.e., $\partial S_{l}=\{ x\in \rr^{n}|\;V(x)=l\}$.
	We call $\partial S_{l}$ as the \textit{level set} of $V(x)$ with \textit{level value} $l$. Under mild conditions, the level set $\partial S_{l}$ is an $(n-1)$-dimension manifold in $\rr^n$ \cite{ref:Chiang:survey1995}.
	
	By the definitions of Lyapunov function and energy function, the minimum of $V(x)$ in $A(x_s)$ is always obtained at $x_s$, which is denoted by $l_{\min}:=V(x_s)$. Generally, there exists a maximal level value, $l_{\max } $, such that $S_l$ remains in $A(x_s)$. That is, for all $l$ satisfying $l_{\min } < l < l_{\max } $, it holds that\footnote{Generally, $S_{l}$ might contain several disjoint connected components, however, there is at most one connected component has nonempty intersection with $A(x_{s})$ \cite{ref:Chiang:survey1995,ref:Chiang:2010}. For the simplicity of denotation, we refer $S_{l}$ in particular to that component, rather than introducing a new symbol.}
	\begin{eqnarray} 
	\label{eq:Sl.2}
	\left\{\begin{array}{l} {S_{l}\cap A(x_{s} ) \ne \emptyset} \\ {\partial S_{l}\cap \partial A(x_{s} )=\emptyset, } \end{array}\right.
	\end{eqnarray} 
	and $\partial S_{l_{\max }}\cap \partial A(x_{s})\neq\emptyset$. It has been shown that $l_{\max}$ is obtained at a type-1 UEP on $\partial A(x_{s})$, which is referred to as the \textit{closest} UEP, denoted by $x_u^{cl}$ \cite{ref:Chiang:1989}. Hence,
	letting $l=l_{\max}=V(x_u^{cl})$, we obtain the maximum global estimation associated with the Lyapunov (energy) function. 
	
	Although $S_{l_{\max}}$ is optimal in the sense of global estimation, it is still too conservative for practical applications. Hence, local estimation methods emerged that use $\partial S_l$ with $l>l_{\max}$ to approximate a local segment of $\partial A(x_s)$. In this case, $\partial S_{l}$ is no longer a subset of $A(x_s)$. The two well-known local methods, i.e., the PEBS method and the BCU method, differ in how they choose $l$. The PEBS method uses the potential energy at an estimated fault-dependent exit point as the level value $l$ \cite{Chiang_Foundationspotentialenergy_1988}. While, in the BCU method, the fault-dependent UEP, referred to as the CUEP, is used to determine the level value $l$ \cite{Chiang_TheoreticalfoundationBCU_1995}. %Since the estimated boundary is related to a specific fault-trajectory, the conservativeness of estimation can be reduced significantly without causing over-optimism. However,  the estimation error could be large if the exit point of the critical fault trajectory is not close to the CUEP.
	
	The following notations will also be used in this paper. For a set $A\subset\rr^n$, $\overline A$ denotes the closure of $A$, $\text{Int}(A):=\overline A \backslash \partial \overline A$ denotes the internal point of $A$, and $A^{c}$ denotes the complement of $A$. For a point $x\in\rr^n$, we define the distance from $x$ to $A$ as 
	$d(x,A):=\inf_{y\in A}\|x-y\|_2$. For any $t\in\rr$, we define
	$\phi^f_{t}(A):=\left\lbrace y\in\rr^n|\;y=\phi^f_{t}(x),x\in A\right\rbrace $.

	\section{ Main Theoretical Results}
	\label{sec:3}
	%	\subsection{Basic Idea}
	In this section, we establish the theoretical foundation for the expansion methodology in both global and local senses. We show that the flow mapping is exactly an expansion operator to improve the estimation. We identify the relation among the exact stability boundary, the initial level set, and the expanded level sets via the flow mapping.
	
	%	Obviously, the stable and unstable manifolds, the stability region, the stability boundary, and $S_{l}$  all are invariant sets. Note that, when $l > l_{\min}$, the set $S_A:=S_{l}\cap A(x_{s})$ is also simply connected and open, which is homeomorphic to $\rr^n$. It implies $S_A$ and $A(x_s)$ are homeomorphic. This fact motivates that, if we can come up with the homeomorphic operation, we will be able to map $S_A$ to $A(x_s)$ homeomorphically. Consequently, the boundary of $S_A$ under such an operation can eventually approach the exact stability boundary $\partial A(x_s)$ with arbitrary accuracy. We claim the flow mapping of \eqref{eq:system.1}  exactly defines such a homeomorphic operation as expected. 
	
	\subsection{Global Properties}
	\label{subsec:result.global}
	We first discuss the properties of the flow mapping on an invariant subset of $A(x_s)$. In this case, the expansion results in a new inner approximation for the entire $\partial A(x_s)$, namely the global case. 
	
	\begin{lemma} \label{lem:expansion.1}
		(\textit{Expansion via the flow mapping}): Let $x_s$ be an ASEP of \eqref{eq:system.1}. Let $\Omega $ be a connected and positively invariant set that satisfies $x_{s} \in \text{Int}(\Omega) \subset A(x_{s})$. Then the following statements hold:
		\begin{enumerate}
			\item $\phi^f_{-t}(\Omega)$ is a connected and positively invariant set for all $t\in(0,+\infty)$;
			\item $\Omega \subset \phi _{-t}^{f} (\Omega )\subset A(x_{s} )$  for all $t\in (0,\,+\infty)$;
			\item $\phi _{-t_1}^{f} (\Omega )\subset \phi _{-t_2}^{f} (\Omega )$ if $0 < t_1 <t_2 <+\infty$;
			\item $\lim _{t\to +\infty } \phi _{-t}^{f} (\Omega )=A(x_{s} ).$% \peng{$\lim _{t\to +\infty } \phi _{-t}^{f} (\Omega )$ is well defined}
			%$\forall x\in A(x_s), \exists \tau>0$ such that $\forall t>\tau, x\in\phi _{-t}^{f} (\Omega )$. \peng{new version}
		\end{enumerate}
	\end{lemma}
	
	This lemma reveals that a connected and positively invariant subset of the stability region can be expanded via the (inverse-time) flow mapping $\phi _{-t}^{f}$. Particularly, any point in $A(x_s)$ can be reached as $t\to +\infty$. 
	
	The next lemma identifies the relationship between the exact stability boundary and the expanded boundaries of a positively invariant subset under the flow mapping. 
	
	\begin{lemma} \label{le:boundary.2}
		(\textit{Properties of the boundaries under the flow mapping}): Let $x_s$ be an ASEP of \eqref{eq:system.1}. Let $\Omega $ be a connected and positively invariant set that satisfies $x_{s} \in \text{Int}(\Omega)\subset A(x_{s})$. Then the following statements hold:
		\begin{enumerate}
			\item $\partial (\phi _{-t}^{f} (\Omega ))= \phi^f_{-t}(\partial \Omega)$  for any $t\in (0,\,+\infty)$.
			\item $\lim_{t\to +\infty }d(x,\phi _{-t}^{f} (\partial \Omega ))=0$ for all $x\in\partial A(x_s)$.
		\end{enumerate}
	\end{lemma}

	%	\begin{figure}[htp]
	%		\centering
	%		\subfigure[$x_0 \in \text{Int}(\Omega)$]{
	%			\includegraphics[width=0.7\columnwidth]{fig1a}}
	%		\subfigure[$x_0 \in \phi^f_{-t}(\Omega) \backslash \overline \Omega$]{
	%			\includegraphics[ width=0.7\columnwidth]{fig1b}}
	%		\caption{Boundaries of positively invariant sets under flow mapping.}
	%		\label{fig.1}
	%		%% label for entire figure
	%	\end{figure}
	
	%	\begin{remark}\peng{new}
	%		From Lemma \ref{lem:expansion.1} we have $\phi _{-t_1}^{f} (\Omega )$ is monotonically increasing as $t$ increases and $\phi _{-t}^{f} (\Omega )\subset A(x_{s} )$  for any $t\in (0,\,+\infty)$. This indicates that $\lim _{t\to +\infty } \phi _{-t}^{f} (\Omega)$ exists when $A(x_s)$ is bounded. In this case, by Lemma \ref{lem:expansion.1} and Lemma \ref{le:boundary.2}, we have $\lim _{t\to +\infty }\phi _{-t}^{f} (\Omega)=A(x_s)$ and $\lim _{t\to +\infty }\phi _{-t}^{f} (\partial\Omega)=\lim _{t\to +\infty }\partial\phi _{-t}^{f}(\Omega)=\partial A(x_s)$.
	%	\end{remark}
	
	Lemma \ref{lem:expansion.1} and Lemma \ref{le:boundary.2} indicates that the expended boundary of a subset of $A(x_s)$ never leaves $A(x_s)$ and hence is a global inner approximation of $\partial A(x_s)$. Moreover, it approaches $\partial A(x_s)$ as $t\to+\infty$. 
	
	Given a Lyapunov (energy) function of \eqref{eq:system.1}, for all $l\in (l_{\min} ,\;l_{\max } )$, $S_l$ is exactly a connected and positively invariant subset of $A(x_s)$. Hence, it can be expanded via the flow mapping globally, as stated in the following theorem. 
	
	\begin{theorem}
		\label{thm:energy surface}
		(\textit{Expansion of level set via the flow mapping}): Consider the system \eqref{eq:system.1} satisfying Assumption \ref{as:2}. For every level value $l$ such that $S_{l}\subset A(x_s)$ and $\partial S_{l} \bigcap \partial A(x_s)=\emptyset$, the following statements hold:
		\begin{enumerate}
			\item $S_{l}\subset \phi _{-t_1}^{f} (S_{l})\subset \phi _{-t_2}^{f} (S_{l})\subset A(x_{s} )$ for all $t_1$ and $t_2$ satisfying $0<t_1<t_2<+\infty $. 
			\item $\lim_{t\to +\infty }d(x,\phi _{-t}^{f} (\partial S_{l}))=0$ for all $x\in\partial A(x_s)$.
		\end{enumerate}	
	\end{theorem}
	\begin{proof}
		It directly follows from Lemma \ref{lem:expansion.1} and Lemma \ref{le:boundary.2}.
	\end{proof}
	
	\begin{example}\label{eg:1}
		For the purpose of illustration, we apply our theory to the reduced three-machine power system, which has been extensively studied as a fundamental example in the literature \cite{ref:Chiang:1989-2,ref:Chiang:1989,ref:Chiang:survey1995}:
		
		\begin{equation} \label{eq:3machine} 
		\left\{\begin{array}{l} {\dot{x}_{1} =-\sin x_{1} -0.5\sin (x_{1} -x_{2} )+0.01} \\ {\dot{x}_{2} =-0.5\sin x_{2} -0.5\sin (x_{2} -x_{1} )+0.05} \end{array}\right.  
		\end{equation} 
		
		An energy function for this system is:
		\begin{equation}\label{eq:EF_3}
		\begin{aligned} 
		V(x_{1},x_{2})&=4.0035-2\cos x_{1} -\cos x_{2}  \\ 
		&-\cos (x_{1} -x_{2} )-0.02x_{1} -0.1x_{2} 
		\end{aligned}
		\end{equation}
		%where the constant, 4.0035, is added for keeping the energy function positive-definite in the entire stability region of the ASEP, $x_{s} =(0.028,0.064)$, and achieving zero at $x_{s} $. 
		Due to the periodicity of \eqref{eq:3machine}, if $x_s=(x_{1} ,x_{2} )$ is an ASEP, then $(x_{1} \pm 2m\pi ,x_{2} \pm 2n\pi )$ is an ASEP as well for all integral $m$ and $n$. Therefore, we constrain the area of interest to $X=\left\{(x_{1} ,x_{2} )|\;-2\pi \le x_{1} \le 2\pi ,-2\pi \le x_{2} \le 2\pi \right\}$.
		
		For the system \eqref{eq:3machine} with the energy function \eqref{eq:EF_3}, suppose the level set $\partial S_1=\{ x\in \rr^{n}|\;V(x)=1\}$ is the initial estimation of stability boundary, which is depicted by solid black line in Fig. \ref{fig:levelset}. The exact stability boundary is depicted by red line. For $t=1$s, $\phi _{-1}^{f} (\partial S_1)$ and $\phi _{-2}^{f} (\partial S_1)$ are shown as the Expansion I and II, respectively in Fig. \ref{fig:levelset}.
		The flow mapping expands the estimation monotonically, and the expanded boundaries are kept within $A(x_s)$, which shows the estimation is improved in a global sense.
		\begin{figure}[htb]
			%\centering
			\centering
			\includegraphics[width=1\columnwidth]{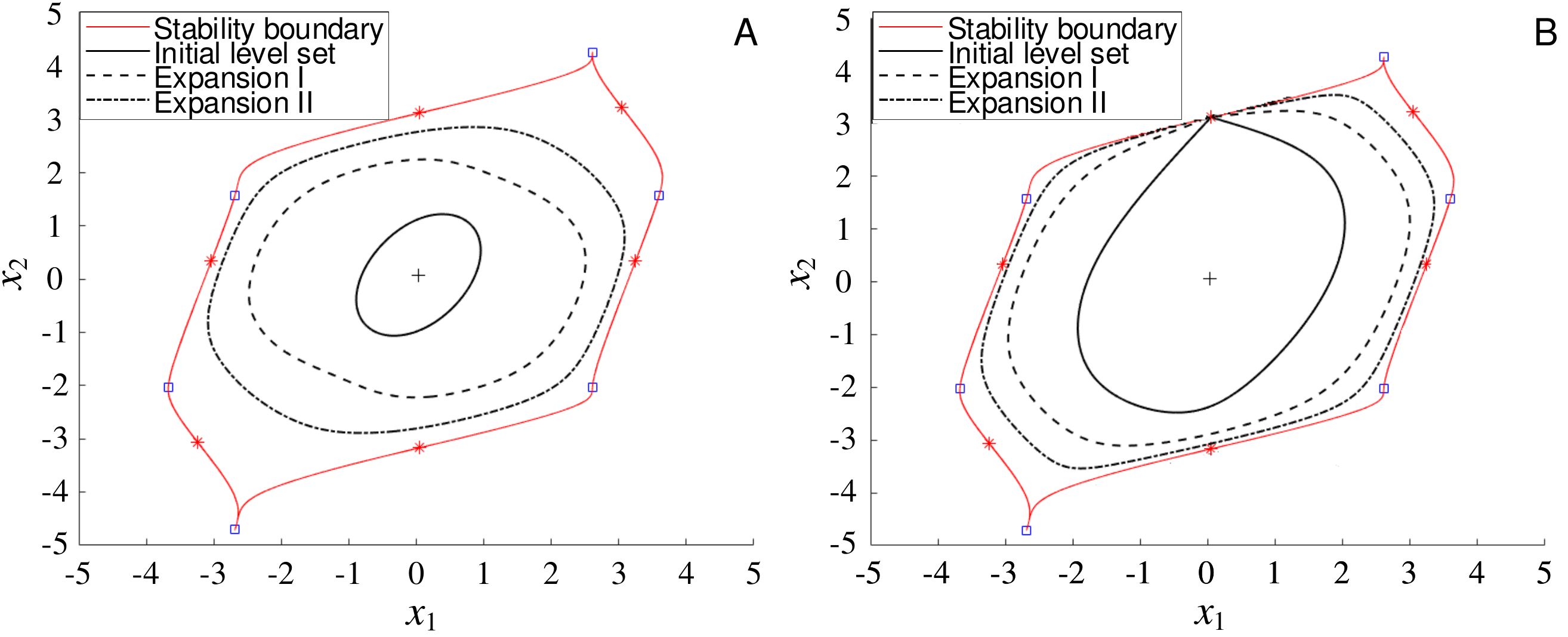}
			\caption{The expansion process of (A) the level set $\partial S_1$ and (B) the level set passing through the closest UEP. Both of them are subsets of the exact stability region.}
			\label{fig:levelset}
		\end{figure}
	\end{example}
	\begin{example}
		We use a 3-dimensional example to show the expansion of level set via the flow mapping. Consider the system in \cite[Example 4.6]{Chesi_DomainAttractionAnalysis_2011} as follows:
		\begin{equation}\label{eq:3d}
		\left\lbrace
		\begin{aligned}
		\dot{x}&=x-x^3+0.5z+y^2\\
		\dot{y}&=-y-y^3+0.5z^2-x^2\\
		\dot{z}&=x+2y-z^3+x^2-y^2
		\end{aligned} \right. 
		\end{equation}
		The system has an ASEP $(x_s,y_s,z_s)=(1.367,-0.849,0.936)$.
		A Lyapunov function of \eqref{eq:3d} is:
		\begin{equation}
		V(x,y,z)=(x-x_s)^2+(y-y_s)^2+(z-z_s)^2
		\end{equation}
		In \cite{Chesi_DomainAttractionAnalysis_2011}, the level set $\partial S_l$ with level value $l=0.291$ is used as an estimation of the stability boundary, which we depict by the red surface in Fig. \ref{fig:3d}. For $t=0.05$s, the results of the flow mapping $\phi _{-t}^{f} (\partial S_l)$ and $\phi _{-2t}^{f} (\partial S_l)$, namely the Expansion I and II, are also showed by the green and gray surfaces in Fig. \ref{fig:3d}, respectively.
		\begin{figure}[htb]
			\centering
			\includegraphics[width=0.61\columnwidth]{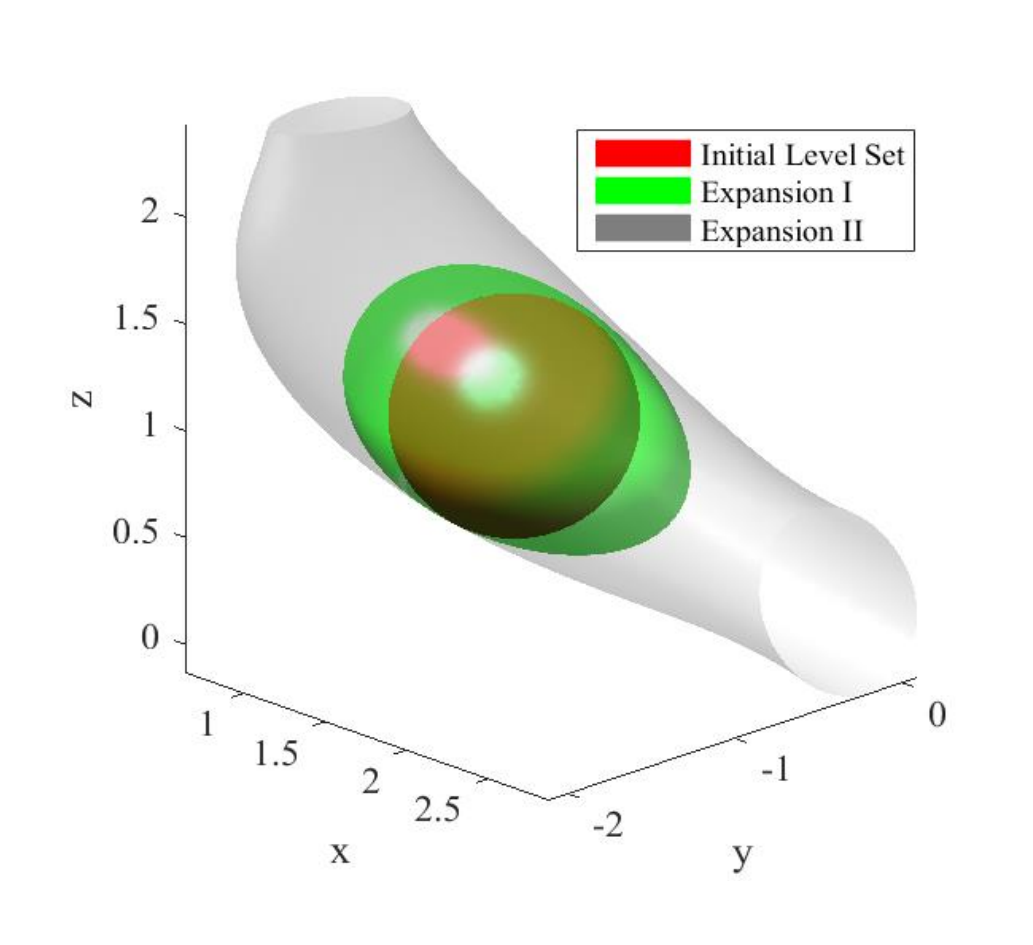}
			\caption{The expansion process in a 3-dimensional system.}
			\label{fig:3d}
		\end{figure}
	\end{example}
	
	Specifically, the level value $l$ can be chosen as the value at the closest UEP $x_u^{cl}$ on the stability boundary. The next theorem shows that $x_u^{cl}$ is the fixed point of the expansion and the global property still holds.
	
	\begin{theorem}
		\label{thm:closestUEP}
		(\textit{Expansion of level set passing through the closest UEP}): Consider the system \eqref{eq:system.1} satisfying Assumption \ref{as:2}. Let $l=V(x^{cl}_u)$. Then the following statements hold:
		\begin{enumerate}
			\item $x^{cl}_u\in \phi^f_t(\partial S_{l})$ for all $t\in \rr$.
			\item $S_{l}\subset \phi _{-t_1}^{f} (S_{l} )\subset \phi _{-t_2}^{f} (S_{l})\subset A(x_{s} )$ for all $t_1$ and $t_2$ satisfying $0<t_1<t_2<+\infty $. 
			\item $\lim_{t\to +\infty }d(x,\phi _{-t}^{f} (\partial S_{l}))=0$ for all $x\in\partial A(x_s)$.
		\end{enumerate}	
	\end{theorem}

	The level set passing through the closest UEP gives the largest possible global estimation with the given Lyapunov (energy ) function. Theorem \ref{thm:closestUEP} indicates that such an estimation can be expanded by the flow mapping and approach the exact boundary, while remains being a valid global estimation for all time.
	
	\begin{example}\label{eg:2}
		To illustrate the expansion process based on the closest UEP, we consider the same three-machine system \eqref{eq:3machine}. Calculation shows that the type-1 UEP (0.04667, 3.11489) has the lowest energy value among all the UEPs on the stability boundary. Hence, $x_u^{cl}=(0.04667, 3.11489)$ and $V(x_u^{cl})=3.777$. Let $l=3.777$, $\partial S_{l}$ is the initial estimation of $\partial A(x_{s})$, which is depicted by the solid black line in Fig. \ref{fig:levelset}(b). For $t=1$s, the results of $\phi _{-t}^{f} (\partial S_1(x))$ and $\phi _{-2t}^{f} (\partial S_1(x))$ are shown in Fig. \ref{fig:levelset}(B).
		%		\begin{figure}[htp]
		%			\centering
		%			\includegraphics[width=0.8\columnwidth]{image_closest}
		%			\caption{The expansion process based on closest UEP.}
		%			\label{fig:levelset}
		%		\end{figure}
		It shows that the expansion significantly improved the estimation in a global sense and the closest UEP remains fixed during expansions. This result is consistent with Theorem \ref{thm:closestUEP}.
	\end{example}
	
	Theorem \ref{thm:energy surface} and Theorem \ref{thm:closestUEP} reveal that if a level set is completely contained in $A(x_s)$, then it can serve as a valid initial guess and can be expanded to approach $\partial A(x_s)$ globally without causing over-optimistic estimation. Nevertheless, such global property does not hold when the initial level set is not a subset of $A(x_s)$. In this case, the next subsection will show that the expansion via flow mapping can still improve the estimation locally in the sense that part of $\partial S_l$ approaches part of $\partial A(x_s)$ from within. 
	
	\subsection{Local Properties}
	\label{subsec:result.local}
	%	\begin{lemma}
	%		\label{lem:intersection}
	%		Assume $A_{1} $ and $A_{2} $ are two connected and positively invariant sets of system \eqref{eq:system.1}. Then $A:= A_{1} \cap A_{2} $ is also connected and positively invariant along the trajectories of system \eqref{eq:system.1}.  
	%	\end{lemma}
	%	\begin{proof}
	%		 Choose two arbitrary points, $x_{1} ,x_{2} \in A=A_1\cap A_2 $. Then from the simple connectivity of $A_{1} $ and $A_{2}$, it can be concluded that $A$ is also connected. Similarly, the positive invariance of $A$ along the trajectories of system \eqref{eq:system.1} can also be proved. 
	%	\end{proof}
	
	We start with the following lemma that characterizes the expansion of a level set passing through an arbitrary UEP on the stability boundary. 
	
	\begin{lemma}\label{thm:UEP}
		(\textit{Expansion of the level set passing through an arbitrary UEP on the stability boundary}): 
		Consider the system \eqref{eq:system.1} satisfying Assumption \ref{as:2}. Let $l=V(x^{b}_u)$, where $x^b_u$ is an arbitrary UEP on $\partial A(x_s)$. Define the subset $
		D_{l}:=S_{l} \cap A(x_{s})$. The following statements hold:
		\begin{enumerate}
			\item $x^{b}_u\in \phi^f_t(\partial D_{l})$ for all $t\in \rr$.
			\item $\phi^f_{-t}(D_{l}) = \phi^f_{-t}(S_{l}) \cap A(x_s)$ for all $t\in (0,+\infty)$. 
			\item $D_{l}\subset \phi _{-t_1}^{f} (D_{l} )\subset \phi _{-t_2}^{f} (D_{l})\subset A(x_{s} )$ for all $t_1$ and $t_2$ satisfying $0<t_1<t_2<+\infty $. 
			\item $\lim_{t\to +\infty }d(x,\phi _{-t}^{f} (\partial D_{l}))=0$ for all $x\in\partial A(x_{s})$.
		\end{enumerate}	
	\end{lemma}
	
	Lemma \ref{thm:UEP} indicates that $\partial D_l$ is a global estimation and tends to $\partial A(x_s)$ via flow mapping expansion. However, since $\partial D_{l}$ consists of both $\partial S_{l} $ and $\partial A(x_{s})$, it is unclear how the level set $\partial S_{l}$ behave via the flow mapping expansion. In this regard, the following theorem reveals the topological characteristics of $\partial S_l$ in the limit under the flow mapping. 
	
	\begin{theorem}
		\label{thm:boundary.3}
		(\textit{Characteristics of the limit set of the expanded level set}): Consider the system \eqref{eq:system.1} satisfying Assumptions \ref{as:1} and \ref{as:2}. Let $l=V(x^{b}_u)$, where $x^b_u$ is an arbitrary UEP on $\partial A(x_s)$. The following statement holds
		\begin{equation*}
		\mathop{\lim}\limits_{t\to +\infty }d(x, \phi _{-t}^{f} (\partial S_{l}\cap\overline{A(x_{s})}))=0,\; \forall x\in\bigcup _{x_{u} \in \partial A(x_{s} )\cap S_{l}^{c}\cap E} \mW^{s} (x_{u}).
		\end{equation*}
	\end{theorem}
	
	Recall Theorem \ref{thm:boundary.1} that $\partial A(x_s)$ consists of all the stable manifolds of all UEPs on the boundary. Theorem \ref{thm:boundary.3} shows that the level set $\partial S_l$ tends to only part of $\partial A(x_s)$, unless $\partial A(x_{s} )\cap S_{l}^{c}\cap E=\partial A(x_{s})\cap E$. It is easy to show that the equality holds only if $x_u^b$ is the closest UEP $x_u^{cl}$, in which case the result of Theorem \ref{thm:boundary.3} is coincident with Theorem \ref{thm:closestUEP}. The segments of $\partial A(x_s)$ that $\partial S_l$ will tends to consists of the stable manifolds of UEPs that are located on $\partial A(x_s)\cap S^c_{l}$. 
	
	%	It hence provides a theoretic foundation for applying the expansion scheme to improve the local estimation methods of stability boundary, such as the CUEP method as stated in the following corollary.
	
	\begin{corollary} \label{cro:cuep}
		(\textit{Expansion of the level set passing through the CUEP.}): Assume system \eqref{eq:system.1} satisfies Assumptions 1 and 2. Let $\partial S_{ct}$ denote the level set passing through the CUEP, $x^{ct}_u$. Then, it holds that
		\begin{eqnarray*}
			\lim _{t\to +\infty } d(x,\phi _{-t}^{f} (\partial S_{ct} ))=0,\;\; \forall x\in\mW^{s} (x^{ct}_u ).
		\end{eqnarray*}
	\end{corollary}  
	
	Corollary \ref{cro:cuep} shows the expansion method can apply to the local boundary estimation from the BCU method. Since the expanded boundary approaches the entire stable manifold of $x^{ct}_u $, the estimation accuracy can still be guaranteed even when the exit point is away from $x^{ct}_u$, which has been the main concerns of the BCU mehtod in applications. Also, note that the set $\partial A(x_s)\cap S^c_{l}$ may contain other UEPs in addition to the CUEP. In such a case, the valid scope of approximate boundary will become larger and more reliable. 
	
	%	It is interesting to re-visit the closest UEP case characterized in Theorem \ref{thm:closestUEP}. Noticing $S_{l} \subset A(x_s)$ and the set $\partial A(x_s)\cap S^c_{l}$ includes all the UEPs on the stability boundary, we immediately have the following corollary: 
	
	%	\begin{corollary} \label{cro:closestUEP}
	%		(\textit{Expansion of the level set passing through the closest UEP.}): Assume system \eqref{eq:system.1} satisfies Assumptions 1 and 2. Denote $\partial S_{l}$ as the level set passing through the closest UEP, $x^{cl}_u$, on the stability boundary. Then under the operation of the flow mapping $\phi _{t}^{f}$, the following statement holds:
	%		\begin{eqnarray*}
	%		\lim_{t\to +\infty }d(x,\phi _{-t}^{f} (\partial S_{l}))=0,\, \forall x\in\partial A(x_s)
	%		\end{eqnarray*}
	%	\end{corollary}  
	%	
	%	This corollary is in coincidence with Theorem \ref{thm:closestUEP}. In this sense, the global properties developed in Section \ref{subsec:result.local} can be recognized as special cases of the local ones.  
	
	\begin{example}
		We use the same system \eqref{eq:3machine} to illustrate the local expansion. Suppose a fault occurs, and the corresponding CUEP is (0.0467,-3.1683). The level set $\partial S_{l}$ passing through the CUEP is depicted by the solid black line in Fig. \ref{fig:CUEP}, serving as the initial local estimation of stability boundary. 
		\begin{figure}[htb]
			\centering
			\includegraphics[width=0.8\columnwidth]{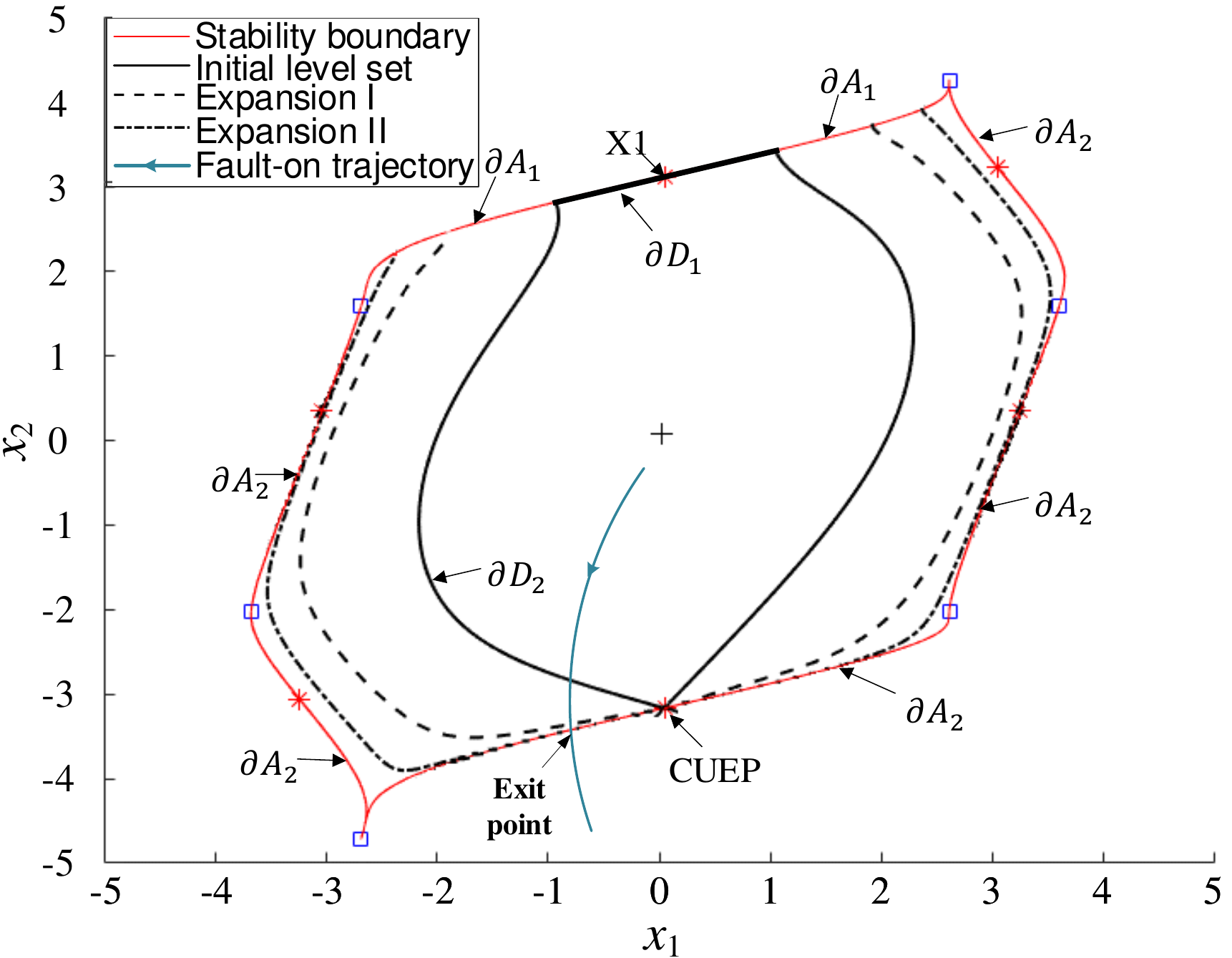}
			\caption{The expansion process of the level set passing through the CUEP.}
			\label{fig:CUEP}
		\end{figure}
		
		Several important sets mentioned in Theorem \ref{thm:boundary.3} and its proof are also marked in this figure. This is a simple case with only one UEP (0.4667, 3.1149) contained in $S_{l} $. Thus, we have that $X_{1}=\left\{\left(0.{\rm 4667},{\rm 3}.{\rm 1149}\right)\right\}$, and $X_{2} $ is composed of other EPs. 
		The flow mapping $\phi_{-t}^f$ is applied to $\partial S_{l}$, with $t=1.5$s. Fig. \ref{fig:CUEP} displays the results of Expansion I and II. It shows that the estimated region expands monotonically and $\partial D_2=\partial S_{ct} \cap \overline{A(x_{s})}$ tends to part of $\partial A_2$ that consists of stable manifolds of UEPs in $X_2$.
	\end{example}

	\section{ Expansion Algorithms}
	\label{sec:application}
	In this section, we propose algorithms to implement our expansion theories. The algorithms rely on a Runge-Kutta scheme to approximate the flow mapping, which will be introduced first. Then, the scheme is applied to level sets, leading to an algorithm that improves the stability boundary estimation. When the CCT for a given fault rather than the stability boundary is of interest, we propose an algorithm to improve the initial CCT estimation, which is more computationally efficient compared with the first one.
	\subsection{ Basic Idea}
	\label{sec:expansion scheme}
	The flow mapping $\phi_t^f$ is the key operator to realize the expansion. However, to obtain the closed-form expression of $\phi_t^f$ essentially requires solving \eqref{eq:system.1} analytically, which is usually impossible for power systems. Alternatively, following the idea of numerical ODE methods, we can approximate $\phi_t^f$ by the well-known Runge-Kutta method.
	
	To calculate the expanded boundary $\phi_{-t}^f(\partial S_l)$ for some $t>0$, it is one way to calculate $\phi_{-t}^f(x)$ for all $x\in\partial S_l$ so that $\phi_{-t}^f(\partial S_l)=\cup_{x\in\partial S_l}\phi_{-t}^f(x)$, which was the idea in the classical TRM \cite{ref:Genesio:1985}. However, it is another and more efficient way that we first refine the Lyapunov (energy) function, and then obtain the expansion by the identity
	$$\phi_{-t}^f(\partial S_l)=\left\lbrace  x\in\rr^n |\; V(\phi_{t}^f(x))=l \right\rbrace.$$
	The latter approach avoids point-wise integration for all $x\in\partial S_l$, which is of infinite-dimension. It also provides an approximate closed-form expression for the boundary, i.e., $V(\phi_{t}^f(x))=l$. In light of these considerations, we propose the following general expansion scheme.% Chiang proposes two algorithms based on symbolic calculation, which are proved theoretically feasible to high-dimensional general systems \cite{ref:Chiang:1989-2,ref:Chiang:1990}. The essences of these two algorithms are the backward Euler integration method and the Trapezoidal integration method, respectively.   
	
	\textbf{General Expansion Scheme}: For a given Lyapunov (energy) function $V(x)$ of system \eqref{eq:system.1}, choose an appropriate step size $h$, and construct the following iterative sequence 
	\begin{equation} \label{eq:ExpansionScheme} 
	\begin{aligned}
	V_{1} (x)&=V(N_{h}^{f} (x))\\
	V_{2} (x)&=V_{1} (N_{h}^{f} (x))\\
	&\vdots\\
	V_{k} (x)&=V_{k-1} (N_{h}^{f} (x)).
	\end{aligned}
	\end{equation} 
	Here, 
	\begin{equation} \label{eq:RK}
	N_{h}^{f} (x)= x+h\sum _{i=1}^{s}b_ik_i,\qquad
	k_i=f(x+h\sum_{j=i}^{i-1} a_{ij}k_j),
	\end{equation} 
	is the explicit $s$-stage Runge-Kutta method \cite[Chapter 1]{iserles_2008} so that $\phi_{h}^f(x)\approx N_{h}^{f} (x)$. The constant coefficients $a_{ij}$  (for $1 \leq j < i \leq s$) and $b_i$ (for $i = 1,\ldots,s$) can be found in the Butcher tableau \cite{butcher1963coefficients}.
	
	We remark that the schemes proposed in the pioneering work by Chiang et al. \cite{ref:Chiang:1989-2,ref:Chiang:1990} can be viewed as special cases of \eqref{eq:RK} with $s=1$, i.e., the Euler method $N_{h}^{f} (x)=x+hf(x)$; or $s=2$, i.e., the improved Euler method
	\begin{equation}\label{eq:RK2}
	N_{h}^{f} (x)=x+h/2(f(x)+f(x+hf(x))).
	\end{equation}
	
	Generally, a higher order of the Runge-Kutta method will result in a smaller error but it also leads to increased computational burden. As shown in Section \ref{sec:5}, we find that the 2-order method \eqref{eq:RK2} is sufficiently accurate meanwhile computationally tractable for symbolic calculations. When only numeric calculations are involved, a 3-order method is still tractable, the formula of which reads
	\begin{equation}\label{eq:RK3}
	\begin{aligned}
	N_{h}^{f}(x)=& x+\frac{h}{6}f(x)+\frac{2h}{3}f\Big(x+\frac{h}{2}f(x)\Big)\\
	&+\frac{h}{6}f\Big(x-hf(x)+2hf\big(x+\frac{h}{2}f(x)\big)\Big).
	\end{aligned}
	\end{equation}
	\subsection{ Improving Stability Boundary Estimation}
	The general expansion scheme can be applied to level sets that estimate the stability boundary. Especially, our theory ensures that the scheme can improve local estimations that are more commonly used in practice compared with global ones, e.g., the BCU and the PEBS methods.  We propose the following Algorithm \ref{al:1} for the implementation.
	
	\begin{algorithm}[htb]
		\caption{Improving Stability Boundary Estimation}
		\begin{algorithmic}[1]\label{al:1}
			\renewcommand{\algorithmicrequire}{\textbf{Input:}}
			\renewcommand{\algorithmicensure}{\textbf{Output:}}
			\REQUIRE $V(x)$, $V_{cr}$, maximum iteration number $M$, step size $h$, and $s$.
			\ENSURE The improved stability boundary estimation.
			\STATE Set $i=1$, $V_0(x)=V(x)$
			\\ \textit{LOOP Process}:
			\WHILE {$i \leq M$}
			\STATE Calculate $V_i(x)=V_{i-1}(N_h^f(x))$ with the $s$-stage scheme \eqref{eq:RK}
			\STATE Let $i=i+1$
			\ENDWHILE
			\RETURN The improved stability boundary estimation:
			\begin{equation*}\label{eq:expandboudary}
			\partial S_M=\left\lbrace  x\in\rr^n |\; V_M(x)=V_{cr} \right\rbrace.
			\end{equation*}
		\end{algorithmic} 
	\end{algorithm}
	
	The BCU method identifies the corresponding CUEP of a fault and uses the level set $\partial S=\left\lbrace  x\in\rr^n|\; V(x)=V_{cr} \right\rbrace$ that passes through the CUEP as a local stability boundary estimation \cite{Chiang_TheoreticalfoundationBCU_1995}. Corollary \ref{cro:cuep} shows that the expansion of $\partial S$ via flow mapping can effectively reduce the conservativeness.
	
	The PEBS method uses the first local maximum of the potential energy along the fault-on trajectory as the critical level value $V_{cr}$. And the level set $\partial S=\left\lbrace  x\in\rr^n|\; V(x)=V_{cr} \right\rbrace$ is a local estimation of the stability boundary.
	We remark that the PEBS method may give conservative or over-optimistic local estimations \cite{Chiang_Foundationspotentialenergy_1988}. For a conservative estimation, $\partial S$ is contained in the real boundary $\partial A$ locally. The expansion result is identical to the BCU situation. However, for over-optimistic estimations, Lemma \ref{thm:UEP} cannot directly apply. Nevertheless, our simulation results shows that the expansion is still effective for all tested PEBS cases. The theory for shrinking over-optimistic estimations is our future task.
	
	%	With a proper order $s$ of algorithm \eqref{eq:ExpansionScheme}, the expansion sequence \eqref{eq:expandboudary} can result in an improved boundary estimation with satisfactory accuracy in a few iterations.
	
	Typically, Algorithm \ref{al:1} requires symbolic calculation. This produces a closed-form estimated boundary $\partial S_k=\left\lbrace  x\in\rr^n|\; V_k(x)=V_{cr} \right\rbrace$, empowered by a closed-form iterative function $V_k(x)$. To our best knowledge, this idea was first introduced in \cite{Chiang_constructivemethoddirect_1988} by Chiang et al. as a constructive methodology. Our results in this paper, further enlarge the valid scope of this idea into the local cases.
	%	\subsection{ Improve stability boundary Estimation of PEBS}
	\label{sec:PEBS}
	\subsection{ Improving CCT Estimation}
	%	CCT is a significant indicator in transient stability analysis of power systems, as stated in Section \ref{sec:1}. In this section, we apply our expansion theory to CCT calculation, proposing a fast CCT estimation algorithm.
	
	Consider a given fault-on trajectory $x_{F}(t)$ starting from the pre-fault initial value $x_{0pre}$. In direct methods, a critical level value $V_{cr}$ is calculated, then the first time when $V(x_F(t))\geq V_{cr}$ is used to estimate the CCT. We assume that the estimated CCT is conservative, i.e., is smaller than the real CCT, which results from the conservative nature of the estimated stability boundary. 
	
	The expansion methodology can effectively reduce the conservativeness and hence improve the CCT estimation.
	Note that, since CCT concerns only one particular fault, the full characterization of the stability boundary is luxuriously exceeding our needs. We only need to evaluate the value of $V(x)$ along the fault-on trajectory $x_F(t)$, and hence avoid symbolic calculations as in Algorithms \ref{al:1}.
	Based on this observation, we propose the following algorithm to improve the CCT estimation.
	\begin{algorithm}[htb]
		\caption{Improving CCT Estimation}
		\begin{algorithmic}[1]\label{al:3}
			\renewcommand{\algorithmicrequire}{\textbf{Input:}}
			\renewcommand{\algorithmicensure}{\textbf{Output:}}
			\REQUIRE $x_F(t)$, $V(x)$, $V_{cr}$, maximum iteration number $M$, iteration step size $h$, and $s$.
			\ENSURE  the improved CCT estimation
			\\ \textit{Initialization}:
			\STATE Find $t_0$ such that $V(x_F(t_0))\geq V_{cr}$ for the first time.
			\STATE Set $i=1$, $x_0=x_F(t_0)$, $V_0(x)=V(x)$.
			\\ \textit{LOOP Process}:
			\WHILE {$i \leq M$}
			\STATE Calculate $V_i(x_F(t))=V_{i-1}(N_h^f(x_F(t)))$ with the $s$-stage scheme \eqref{eq:RK} for $t\geq t_{i-1}$, until $t=t^*$ such that $V_i(x_F(t^*))\geq V_{cr}$ for the first time. 
			\STATE Let $t_i=t^*$, $i=i+1$.
			\ENDWHILE
			\RETURN $CCT=t_M$ 
		\end{algorithmic} 
	\end{algorithm}
	
	Since the flow mapping here only acts on the fault-on trajectory, the iteration $V_i(x)=V_{i-1}(N_h^f(x))$ can be executed point-wise by numeric calculation effectively.
	Geometrically, Algorithm \ref{al:3} expands the estimated boundary $\{x\in\rr^n|V_i(x)=V_{cr}\}$ whenever the fault-on trajectory hits it, i.e., $V_i(x_F(t^*))=V_{cr}$. In another word, the expansion process is triggered only when it is in need. This event-trigger-like process is illustrated in Fig. \ref{fig:al3} Section \ref{sec:5}. All these properties contribute to accelerate the algorithm.
	
	%	We remark that Algorithm \ref{al:3} only applies to conservative initial estimations because an increasing order of search is used in the loop process. %Although a decreasing order can be adopted for over-optimistic situations, it is not a practical method because we cannot predict in advance whether the result is conservative or optimistic.
	\section{ Case Studies}\label{sec:5}
	\subsection{ Simulation Setup}
	In this section, we apply the proposed algorithms to the IEEE 39-bus power system \cite{ref:hiskens2013ieee}, as shown in Fig.\ref{fig:IEEE39}. 
	We first study the zero transfer conductance (lossless) model, which permits a rigorous energy function, and hence our theories apply. We then study the non-zero transfer conductance (lossy) system model, whose energy function is believed to be path-dependent so that approximation of energy function is used. Our simulation results show that the expansion method is still effective in lossy models, while some fluctuation in the process is observed.
	\begin{figure}[htp]
		\centering
		\includegraphics[width=.75\columnwidth]{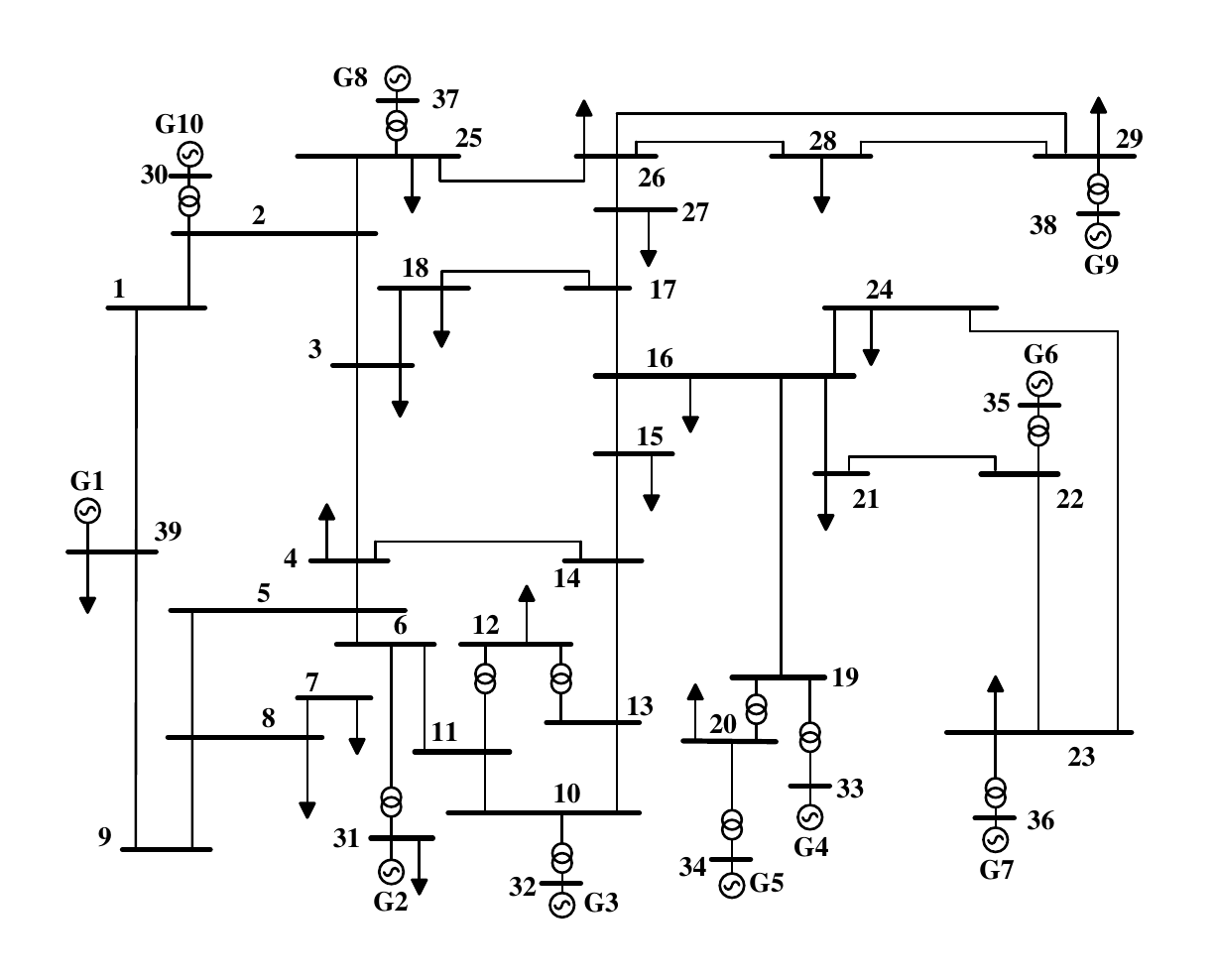}
		\caption{IEEE 39-bus benchmark system.}
		\label{fig:IEEE39}
	\end{figure}
	
	For simplicity, we consider the classical second-order generator model \cite{ref:kundur1994} in transmission networks. Note, however, our method is applicable to more general systems as long as they permit Lyapunov (energy) functions (see \cite{8890862} for example where Lyapunov functions for a power system with heterogeneous devices were proposed). For lossy model, a linear trajectory in the angle space is assumed to construct the path-dependent energy function, which is a common approximation in the literature \cite{ref:Athay_PracticalMethodDirect_1979,ref:Chiang:survey1995,ref:Chiang:2010}. The most commonly used energy function of this power system reads
	\begin{equation}\label{eq:EF}
	\begin{aligned}
	V(\tilde{\omega},\theta)&=\frac{1}{2}\sum_{i=1}^{n}M_i\tilde{\omega}_i^2-\sum_{i=1}^nP_i(\theta_i-\theta_i^s)\\
	&-\sum_{i=1}^{n-1}\sum_{j=i+1}^{n}E_iE_j\bigg[ B_{ij}\left( \cos(\theta_{ij})-\cos(\theta_{ij}^s)\right) \\
	&-G_{ij}\frac{\theta_i+\theta_j-\theta_i^s-\theta_j^s}{\theta_i-\theta_j-\theta_i^s+\theta_j^s}\left( \sin(\theta_{ij})-\sin(\theta_{ij}^s)\right) \bigg] 
	\end{aligned}
	\end{equation}
	For lossless model, the energy function is rigorous and is the same as \eqref{eq:EF} with $G_{ij}=0,~\forall i\neq j$. See \cite{ref:kundur1994,ref:Athay_PracticalMethodDirect_1979,ref:hiskens2013ieee} for more information about the system model and energy functions.
	
	In our simulation, 6 three-phase short circuit faults at 6 different buses are carried out (see Table \ref{tab:CCT_lossless} for faulted bus numbers). For each fault, we first execute the step-by-step (SBS) trajectory simulation to obtain the accurate CCTs and the real exist points, i.e., the intersecting points of the stability boundary and the fault-on trajectories. Then for each fault-on trajectory, we use the dynamic gradient method \cite{Scruggs_Dynamicgradientmethod_2001} to estimate the exit point and the shadowing method \cite{ref:shadowing} to locate the CUEP, based on which the PEBS and the BCU methods are carried out. As a result, we obtain the critical level value $V_{cr}$ for each fault, and consequently the estimated local stability region boundaries $\partial S=\left\lbrace  x\in\rr^n|\; V(x)=V_{cr} \right\rbrace$ and the estimated CCTs.
	Finally, we execute Algorithms \ref{al:1}-\ref{al:3} to improve the boundary and CCT estimations. The following subsections report our results.
	\subsection{ Results of Lossless Model}
	\subsubsection{Improve stability boundary estimation}
	We apply Algorithm \ref{al:1} to the level set $\partial S=\left\lbrace  x\in\rr^n|\; V(x)=V_{cr} \right\rbrace$ obtained from the BCU method and the PEBS method, respectively, with $h=0.2$, $M=9$, and scheme \eqref{eq:RK2}. Each fault-on trajectory intersects the stability boundary at the real exit point and intersects the estimation at another point. We use the Euclidean distance between those two points to measure the difference between the real stability boundary and its estimation, and hence to quantify the effectiveness of the expansion. 
	
	Fig. \ref{fig:dist_lossless} reports the results of six different faults and their average. It shows that our algorithm expands the boundary estimation monotonically towards the real stability boundary, which significantly reduces the conservativeness of both the BCU and the PEBS methods.
	\begin{figure}[htp]
		\centering
		\includegraphics[width=1\columnwidth]{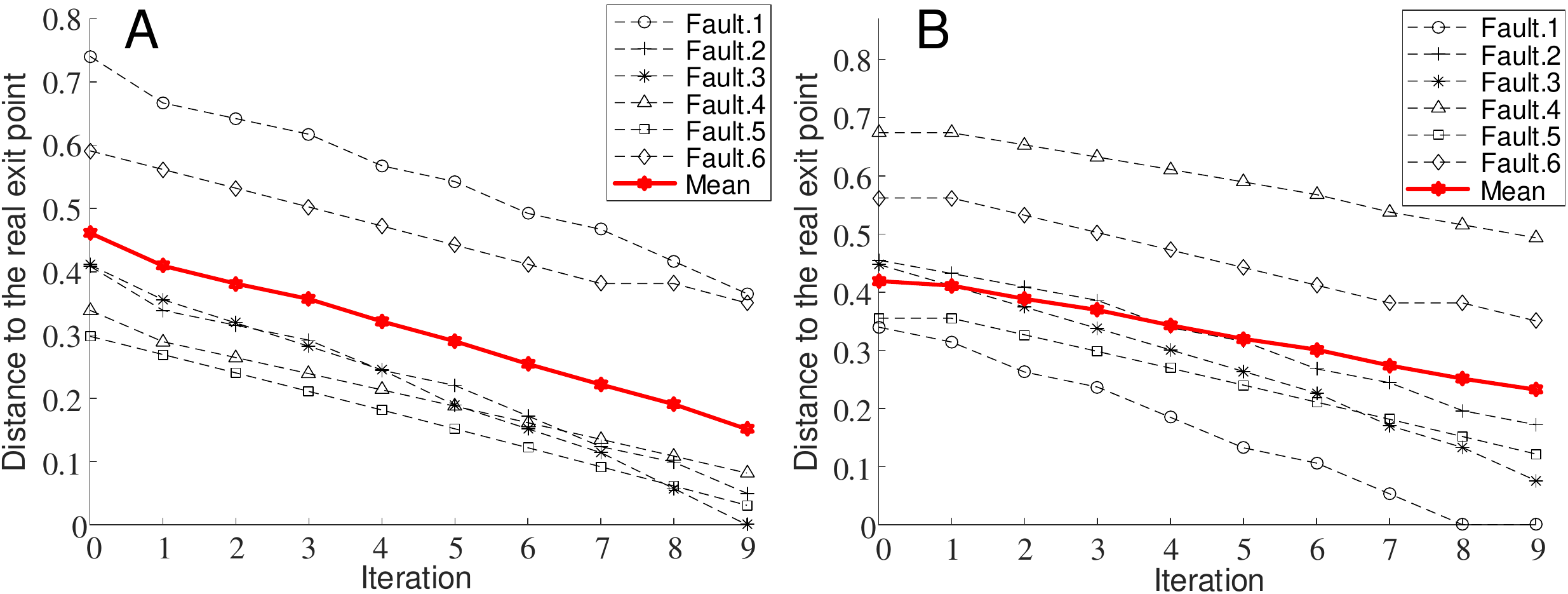}
		\caption{Lossless case: distance between the real exit point and the estimated ones resulting from (A) the BCU method and (B) the PEBS method.}
		\label{fig:dist_lossless}
	\end{figure}

	\subsubsection{Improve CCT estimation}
	We now use Algorithm \ref{al:3} to improve the CCT estimation obtained from the BCU method, with $h=0.2$ and $M=6$. As only numeric calculations are involved, we use the 3-order scheme \eqref{eq:RK3} in the case. Six iterations in total are carried out and the results of iteration 2, 4, and 6 are reported in Tab. \ref{tab:CCT_lossless}. 
	
	\begin{table*}[htb]
		\renewcommand{\arraystretch}{1.3}
		\centering
		\caption{\\ \textsc { Lossless Case: Improving CCT Estimation}}
		\label{tab:CCT_lossless}
		\begin{threeparttable}
	
		\begin{tabular}{|c|ccccccccccccccc|}
			\hline
			\multirow{4}{*}{\bfseries Fault} &
			\multicolumn{2}{c|}{\multirow{2}{*}{\bfseries Step-by-Step}} &
			\multicolumn{4}{c|}{\multirow{2}{*}{\bfseries BCU}} &
			\multicolumn{9}{c|}{\bfseries Expansion}\\
			\cline{8-16}
			& \multicolumn{2}{c|}{} & \multicolumn{4}{c|}{} & \multicolumn{3}{c|}{N=2} & \multicolumn{3}{c|}{N=4} & \multicolumn{3}{c|}{N=6} \\
			\cline{2-16}
			\bfseries Bus& \multirow{2}{*}{Vcr} & \multirow{2}{*}{CCT} & \multirow{2}{*}{Vcr} & \multirow{2}{*}{CCT} & \multirow{2}{*}{Error\tnote{$\intercal$}} & \multirow{2}{*}{T.C.\tnote{$\dagger$}}  & \multirow{2}{*}{CCT} & \multirow{2}{*}{Error} & T.C.\tnote{$\ddagger$} & \multirow{2}{*}{CCT} & \multirow{2}{*}{Error} & T.C. & \multirow{2}{*}{CCT} & \multirow{2}{*}{Error} & T.C.\\
			&&&&&(\%)&(s)&&(\%)&INC(s)&&(\%)&INC(s)&&(\%)&INC(s)\\
			\hline
			3&	23.42	&	0.669	&	16.15	&	0.519	&	-22.40	&	2.156	&	0.528	&	-21.15	&	0.047	&	0.532	&	-20.53	&	0.114	&	0.540	&	-19.28	&	0.269	\\
			\rowcolor{gray!20}9&	19.79	&	0.690	&	18.08	&	0.653	&	-5.40	&	1.376	&	0.662	&	-4.19	&	0.031	&	0.670	&	-2.99	&	0.114	&	0.674	&	-2.38	&	0.214	\\
			14&	20.37	&	0.874	&	18.08	&	0.804	&	-8.00	&	1.292	&	0.816	&	-6.57	&	0.041	&	0.829	&	-5.14	&	0.145	&	0.837	&	-4.19	&	0.280	\\
			\rowcolor{gray!20}20&	9.90	&	0.478	&	8.52	&	0.444	&	-7.10	&	0.479	&	0.452	&	-5.35	&	0.030	&	0.456	&	-4.48	&	0.094	&	0.460	&	-3.61	&	0.192	\\
			%			25&	20.45	&	0.452	&	10.15	&	0.323	&	-28.56	&	4.637	&	0.331	&	-26.72	&	0.043	&	0.331	&	-26.72	&	0.092	&	0.336	&	-25.79	&	0.208	\\
		31&	23.21	&	0.454	&	21.89	&	0.435	&	-4.03	&	1.307	&	0.444	&	-2.19	&	0.033	&	0.448	&	-1.27	&	0.098	&	0.452	&	-0.35	&	0.199	\\
			\rowcolor{gray!20}39&	20.46	&	0.628	&	18.01	&	0.587	&	-6.44	&	1.618	&	0.595	&	-5.11	&	0.020	&	0.600	&	-4.45	&	0.065	&	0.604	&	-3.78	&	0.135	\\

			\hline
			AVE & \multicolumn{4}{c}{} & -8.89 &1.37 & &-7.43& 0.03 & &-6.48& 0.11 & &-5.60 &0.21\\
			
			VAR & \multicolumn{4}{c}{} & 6.17 &0.50 & &6.28& 0.01 & &6.41& 0.02 & &6.25 &0.05\\
			\hline
		\end{tabular}
	\begin{tablenotes}
		\item[$\intercal$] The relative error of CCT estimation, as defined in \eqref{eq:error}.
		\item[$\dagger$] The time consumption of the BCU method.
		\item[$\ddagger$] The time consumption increase caused by the expansion algorithm.
	\end{tablenotes}
		\end{threeparttable}
	\end{table*}
	To demonstrate the effectiveness of our algorithm, we calculate the relative error of CCT estimation, as defined in \eqref{eq:error}, and record the time consumption increase (T.C. INC) of every expansion process.
	\begin{equation}\label{eq:error}
	\text{Error}=\frac{\text{Estimated CCT}-\text{Real CCT}}{\text{Real CCT}}\times 100\%
	\end{equation}
	
	Tab. \ref{tab:CCT_lossless} also reports the arithmetic averages and the standard deviations of the Error and the T.C. INC to illustrate the average performance. It shows that the algorithm monotonically improved the estimation, which is consistent with our theories. On average, six expansions reduce the estimation error by 3.22\% and increase the time consumption by 0.21s, which is 11.41\% of the BCU method.
	It is also noticed that the time consumption of the expansion process varies little for different faults while the time consumption of the BCU method is highly relevant to faults.
	\subsection{ Results of Lossy Model}
	This subsection studies the expansion method in the lossy case, where the energy function is inaccurate. Simulation results show that the expansion method is still effective in this case, while some fluctuation in the process is observed.
	\subsubsection{Improve stability boundary estimation}
	We apply Algorithm \ref{al:1} to the boundary estimations obtain from the BCU method and the PEBS method. All simulation settings are the same as in the lossless case. 	
	\begin{figure}[htp]
		\centering
		\includegraphics[width=1\columnwidth]{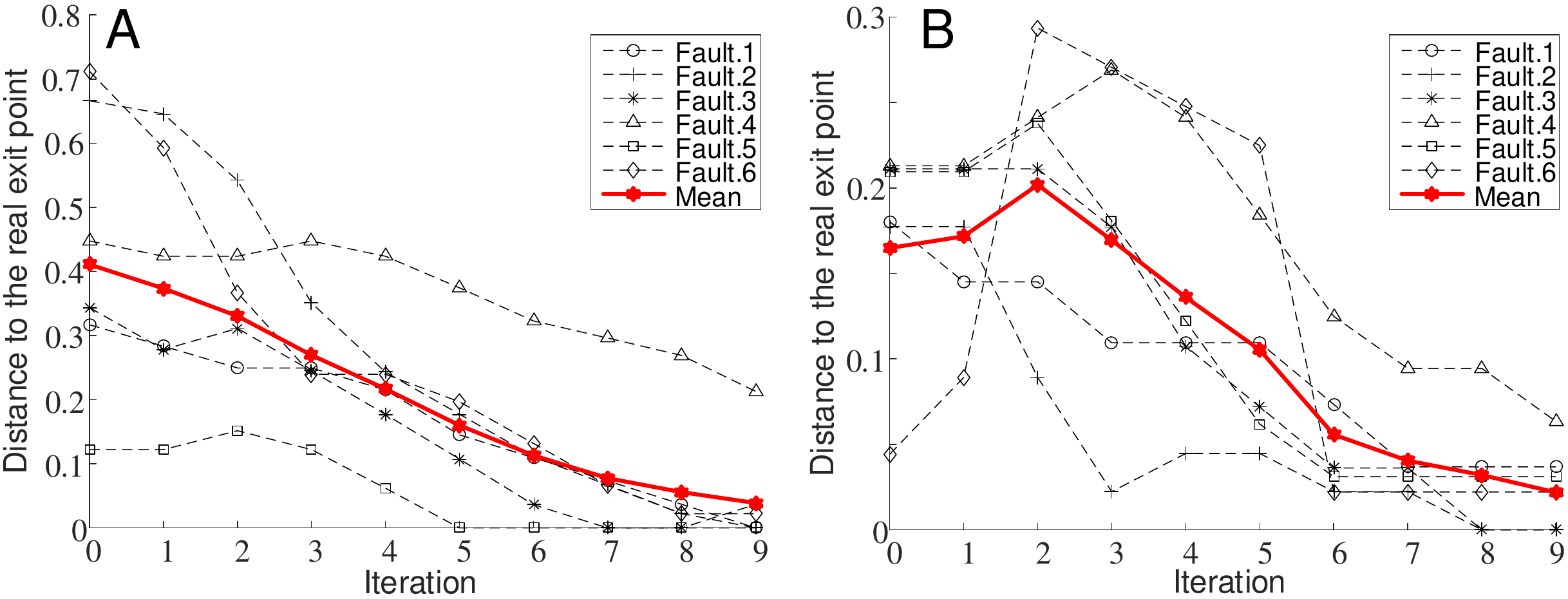}
		\caption{Lossy case: distance between the real exit point and the estimated ones resulting from (A) the BCU method and (B) the PEBS method.}
		\label{fig:dist_lossy}
	\end{figure}
	
	Fig. \ref{fig:dist_lossy} reports the expansion results. The distance indicators fluctuate in the first few iterations, rather than a monotonic improvement as the lossless case, which is a consequence of the inaccurate energy function. Despite that, the distances gradually shrink and tend to zero, which indicates that the estimated boundaries are expanded and converge to the real ones. %Generally, the distance reduces faster for the first few iterations and get slower as the number of iteration increase.
	It is also noticed that fluctuations in the PEBS case are more violent than the BCU case, possibly because the initial PEBS result is more close to the real boundary, and hence it is more sensitive for the inaccuracy of energy function and the computational error. %Mathematically, for a fix step size $t$, the speed of expansion is determined by the norm of the vector field, namely $\lVert f(x)\rVert$ of \eqref{eq:system.1}. %The large jump for the first iteration results from a large $\lVert f(x)\rVert$ at the initial surface $\partial S$. The norm of the vector field at every hitting point on the expanded boundary is displayed in Fig \ref{fig:4}, where the same trend can be seen obviously. 
	
	\subsubsection{Improve CCT estimation}	
	Algorithm \ref{al:3} combined with the BCU method is also implemented for the lossy case to improve the CCT estimation. The results of iteration 2, 4, and 6 are displayed in Tab. \ref{tab:CCT}. 
	
	\begin{table*}[htb]
		\renewcommand{\arraystretch}{1.3}
		\centering
		\caption{\\ \textsc { Lossy Case: Improving CCT Estimation}}
		\label{tab:CCT}
		\begin{threeparttable}
		\begin{tabular}{|c|ccccccccccccccc|}
			\hline
			\multirow{4}{*}{\bfseries Fault} &
			\multicolumn{2}{c|}{\multirow{2}{*}{\bfseries Step-by-Step}} &
			\multicolumn{4}{c|}{\multirow{2}{*}{\bfseries BCU}} &
			\multicolumn{9}{c|}{\bfseries Expansion}\\
			\cline{8-16}
			& \multicolumn{2}{c|}{} & \multicolumn{4}{c|}{} & \multicolumn{3}{c|}{N=2} & \multicolumn{3}{c|}{N=4} & \multicolumn{3}{c|}{N=6} \\
			\cline{2-16}
			\bfseries Bus& \multirow{2}{*}{Vcr} & \multirow{2}{*}{CCT} & \multirow{2}{*}{Vcr} & \multirow{2}{*}{CCT} & \multirow{2}{*}{Error\tnote{$\intercal$}} & \multirow{2}{*}{T.C.\tnote{$\dagger$}}  & \multirow{2}{*}{CCT} & \multirow{2}{*}{Error} & T.C.\tnote{$\ddagger$} & \multirow{2}{*}{CCT} & \multirow{2}{*}{Error} & T.C. & \multirow{2}{*}{CCT} & \multirow{2}{*}{Error} & T.C.\\
			&&&&&(\%)&(s)&&(\%)&INC(s)&&(\%)&INC(s)&&(\%)&INC(s)\\
			\hline
			3&	13.18	&	0.281	&	9.60	&	0.239	&	-14.89	&	1.653	&	0.248	&	-11.93	&	0.013	&	0.252	&	-10.45	&	0.030	&	0.264	&	-6.00	&	0.063	\\
			\rowcolor{gray!20}9&	13.16	&	0.713	&	9.60	&	0.578	&	-18.94	&	0.114	&	0.603	&	-15.44	&	0.015	&	0.661	&	-7.25	&	0.052	&	0.686	&	-3.75	&	0.085	\\
			14&	13.16	&	0.287	&	9.60	&	0.239	&	-16.73	&	0.746	&	0.243	&	-15.28	&	0.007	&	0.260	&	-9.47	&	0.028	&	0.277	&	-3.66	&	0.060	\\
			\rowcolor{gray!20}20&	6.51	&	0.221	&	3.08	&	0.147	&	-33.65	&	0.064	&	0.151	&	-31.77	&	0.009	&	0.151	&	-31.77	&	0.025	&	0.168	&	-24.23	&	0.067	\\
			%		25&	13.19	&	0.242	&	9.78	&	0.205	&	-15.37	&	0.219	&	0.205	&	-15.37	&	0.008	&	0.218	&	-10.21	&	0.038	&	0.226	&	-6.77	&	0.075	\\
			31&	8.72	&	0.284	&	7.41	&	0.259	&	-8.64	&	0.097	&	0.255	&	-10.11	&	0.005	&	0.268	&	-5.71	&	0.029	&	0.276	&	-2.77	&	0.057	\\
			\rowcolor{gray!20}39&	13.85	&	0.844	&	9.60	&	0.696	&	-17.59	&	0.085	&	0.767	&	-9.21	&	0.035	&	0.792	&	-6.24	&	0.056	&	0.813	&	-3.78	&	0.089	\\
			
			\hline
			AVE & \multicolumn{4}{c}{} & -18.41 &0.46 & &-15.62& 0.01 & &-11.82& 0.04 & &-7.37 &0.07\\
			VAR & \multicolumn{4}{c}{} & 7.57 &0.59 & &7.59& 0.01 & &9.08& 0.01 & &7.61 &0.01\\
			\hline
		\end{tabular}
		\begin{tablenotes}
		\item[$\intercal$] The relative error of CCT estimation, as defined in \eqref{eq:error}.
		\item[$\dagger$] The time consumption of the BCU method.
		\item[$\ddagger$] The time consumption increase caused by the expansion algorithm.
	\end{tablenotes}
\end{threeparttable}
	\end{table*}

	It shows that the algorithm effectively improves the CCT estimation for every fault, verifying its capacity in lossy system models. On average, six expansions reduce the estimation error by 10.69\% while increase the time consumption by only 0.07s, which is 16.18\% of the BCU method.
	
	Taking Fault No.1 for example, the process of Algorithm \ref{al:3} is demonstrated in Fig. \ref{fig:al3} to explain its mechanism. The BCU method first provides a $V_{cr}$ and a CCT estimation. Then Algorithm \ref{al:3} is applied, resulting in an iterative expansion process, which is similar to the well-known shadowing method \cite{ref:shadowing}. After six expansions, a much more accurate CCT estimation is produced. We remark that accuracy is at the expense of additional time consumption. The basic idea of Algorithm \ref{al:3} is to improve the estimation by only necessary numeric calculations such that time consumption increase is reduced as much as possible.
	
	\begin{figure}[htp]
		\centering
		\includegraphics[width=.8\columnwidth]{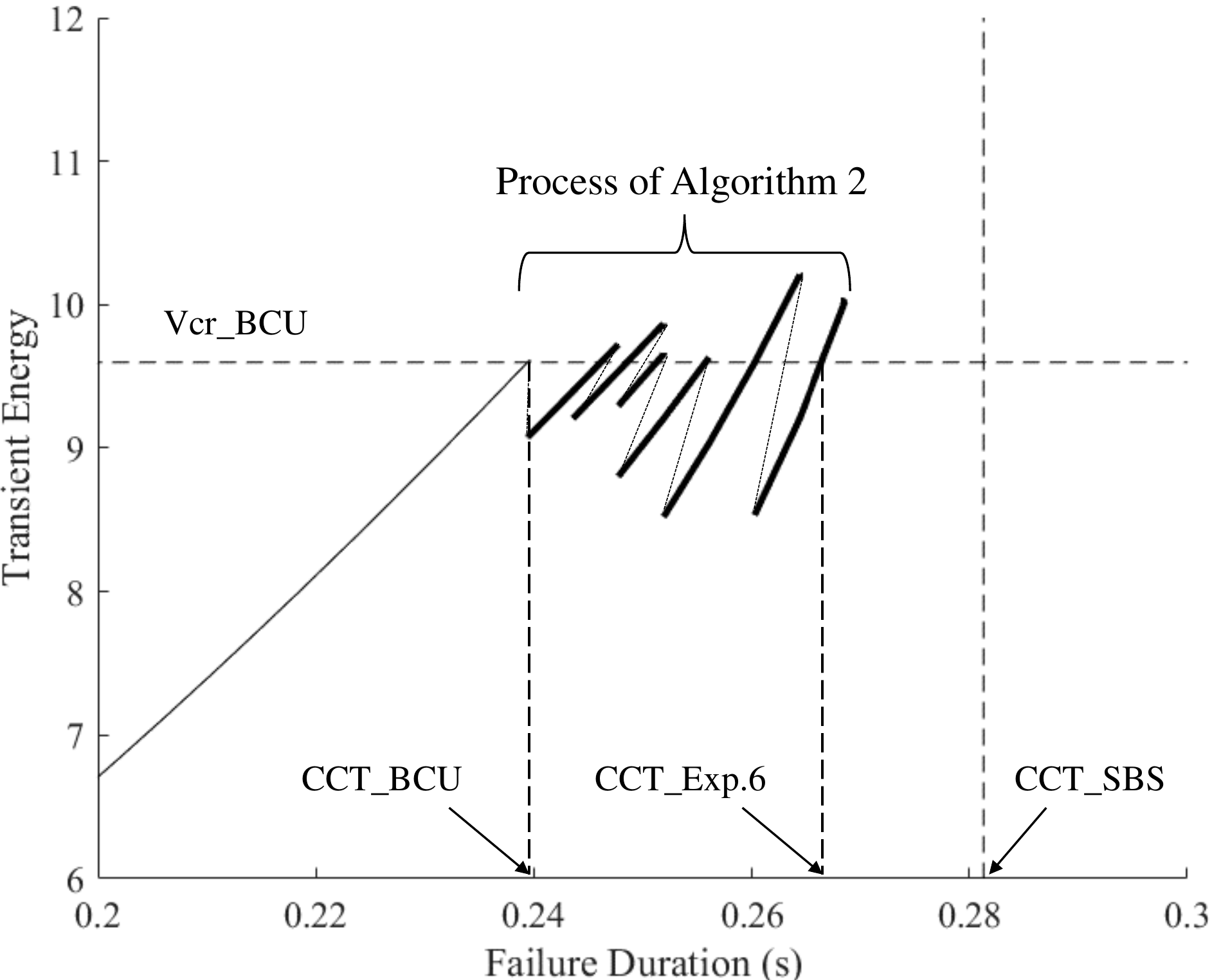}
		\caption{The process of Algorithm \ref{al:3} in Fault No.1 of the lossy case.}
		\label{fig:al3}
	\end{figure}

	In summary, all the simulation results lead to the following conclusions. First, for the lossless case, the accuracy of every estimation is improved monotonically by our algorithms, which validates our expansion theories. Second, for the lossy case, despite some fluctuations during the expansion process mainly due to the inaccurate energy function, the algorithms are able to produce better estimations eventually. This verifies the feasibility and effectiveness of our algorithms for application in lossy systems. Third, in terms of time consumption, Algorithm \ref{al:3} can improve the CCT estimation at a fast speed, which on average can provide 10\% accuracy improvement with just 16\% additional time.

	%	The simulation is carried out by MATLAB R2015b in a PC with Intel i5-7300HQ CPU @2.50GHz @2.50GHz.
	\section{ Concluding Remarks}
	
	In this paper, we have established the theoretic foundation for the expansion methodology via the flow mapping, especially for the local case where the initial guess is only partly contained in the stability region. By leveraging the diffeomorphism of the flow mapping, we have characterized the relation among the initial level set, the expanded level sets, and the limit set under the flow mapping, showing the exact stability boundary can be approached.   
	
	For practical implementation, we have proposed algorithms to improve the stability boundary and the CCT estimations given by direct methods such as the BCU and the PEBS methods. Case studies on the IEEE 39-bus power system benchmark have verified the effectiveness of our algorithms and have shown that remarkable improvement can be achieved with modest computation costs.
	
	It is noticed that our expansion method is still effective even with an inaccurate initial energy function as in the lossy case. This inspires us to further exploit the valid scope of the expansion methodology into systems where only approximate or even no Lyapunov (energy) function is provided.
	\bibliographystyle{IEEEtran}
	\bibliography{mybib}
	
	\appendices
	\makeatletter
	\@addtoreset{equation}{section}
	\@addtoreset{theorem}{section}
	\makeatother
	\renewcommand{\theequation}{A.\arabic{equation}}
	\renewcommand{\thetheorem}{A.\arabic{theorem}}
	\section{Proofs}
	Consider the flow mapping of \eqref{eq:system.1}. Without causing ambiguity in the context, we use $\phi^f_{t} \circ x$ and $\phi^f_{t}(x)$ exchangeably in the Appendix. The following property of the flow mapping will be frequently used in the proofs.
	\begin{property} \label{Prop:1}
		\cite{Kelley_TheoryDifferentialEquations_2010} The flow mapping of system \eqref{eq:system.1}, $\phi^f_t: \rr^n\to\rr^n$, has the following properties:
		\begin{enumerate}
			\item For any given $t\in \rr $, it is a diffeomorphic mapping. 	
			\item It defines an additive group with an ``addition" operation ``$\circ$'' on $\rr$, i.e, for any $x\in\rr^n$	and any $t, \tau\in \rr$, there are
			\begin{equation*}
			\begin{aligned}
			&\phi^f_t \circ \phi^f_{\tau} \circ x =\phi^f_{t+\tau} \circ x=\phi^f_{\tau} \circ \phi^f_t \circ x\\
			&\phi^f_t \circ   \phi^f_{-t} \circ x =\phi^f_0 \circ x=x		
			\end{aligned}
			\end{equation*}
		\end{enumerate}
	\end{property}
	\subsection{Proof of Lemma \ref{lem:expansion.1}}
	\begin{proof}
		$\Rightarrow 1)$. According to Property \ref{Prop:1}$, \phi^f_{t}(x)$ is a diffeomorphic map. Hence, $\phi^f_{-t}(\Omega)$ is connected since $\Omega$ is connected. For the positive invariance, we only need to prove 
		$\phi _{t}^{f} \circ \phi _{-t}^{f} (\Omega )\subseteq \phi _{-t}^{f}(\Omega)$ for $\forall t\in (0,+\infty)$. 	Since $\Omega $ is a positively invariant set, it follows that $\phi _{t}^{f} (\Omega )\subseteq \Omega, \; \forall t\in (0,\,+\infty)$, where ``$=$" holds only for $\Omega = \{x_s\}$. By Property \ref{Prop:1}, it holds that 
		\begin{equation*}
		\phi _{t}^{f} (\Omega )\subseteq \Omega
		\;\;\Rightarrow\;\; \phi _{-t}^{f} \circ \phi _{t}^{f} (\Omega )\subseteq \phi _{-t}^{f}(\Omega) \;\;\Rightarrow\;\; \Omega \subseteq \phi _{-t}^{f}(\Omega).
		\end{equation*}
		%		Hence 1) is proved. 	
		
		$\Rightarrow 2)$. Since $x_s \in \text{Int} (\Omega)$, it holds that $x_s \in \Omega$ and $\{x_s\} \ne \Omega$. Thus, $\phi _{t}^{f} (\Omega )\subset \Omega$ holds for all $t\in (0,\,+\infty)$. Then evolving along the inverse time, we have $\Omega \subset \phi _{-t}^{f} (\Omega ),\; \forall t\in (0,\,+\infty)$. By the definition of $A_s(x)$, it holds that $\phi^f_{-t} (\Omega) \subset A(x_s), \; \forall t\in (0,\,+\infty)$.
		%		Hence 2) is proved.
		
		$\Rightarrow 3)$. It is a direct consequence of 2) by simply replacing $\Omega$ in 2) with $\phi^f_{-t_1} (\Omega)$, and $t$ with  $(t_2-t_1)$. 
		
		$\Rightarrow 4)$. By 3), the sequence of sets $\{\phi _{-t}^{f}(\Omega)\}_t$ is monotonically increasing as $t\to+\infty$ and is bounded by $A(x_s)$. This implies the existence of the limit set and yields $$\lim _{t\to +\infty } \phi _{-t}^{f} (\Omega )=\bigcup_{t>0}\phi _{-t}^{f} (\Omega )$$ 
		By 2) we have $\lim _{t\to +\infty } \phi _{-t}^{f} (\Omega )\subseteq A(x_{s} )$. We now prove $\forall x\in A(x_{s} )$, $\exists \tau>0$ such that $x\in\phi _{-\tau}^{f} (\Omega )$.
		
		Noting that $x_{s} \in \text{Int}(\Omega )$, we can always find a neighborhood $B(x_{s} )$ of $x_{s} $ such that $B(x_{s} )\subset \Omega $. Then, for any $x\in A(x_{s} )$, by definition \eqref{eq:As} there exists a sufficiently large $\tau>0 $ such that $\phi _{\tau}^{f} (x)\in B(x_{s} )$, and hence $\phi _{\tau}^{f} (x)\in \Omega $. Therefore, we have $x\in \phi _{-\tau}^{f} (\Omega )\subset\lim _{t\to +\infty } \phi _{-t}^{f} (\Omega )$. By the arbitrary of $x$, we have $A(x_{s} )\subseteq\lim _{t\to +\infty } \phi _{-t}^{f} (\Omega )$, which completes the proof.
	\end{proof}
	\subsection{Proof of Lemma \ref{le:boundary.2}}
	\begin{proof}%\peng{new. Actually, 1) relies only on diffeomorphic mapping}
		%Because $\Omega $ is connected and positively invariant on system \eqref{eq:system.1}, and $\phi^f_t$ is a diffeomorphic operator for any $t>0$, then $\phi _{-t}^{f} (\Omega )$ is also a connected and positively invariant set on 
		
		$\Rightarrow 1)$.We first prove $\partial (\phi _{-t}^{f} (\Omega )) \subseteq \phi _{-t}^{f} (\partial \Omega )$.
		
		We assume for the purpose of contradiction that there exist $x_0 \notin \partial \Omega$ and $\phi^f_{-t}(x_0) \in \partial (\phi^f_{-t}(\Omega))$ for some finite $t>0$. 
		
		Since $x_0 \notin \partial \Omega$, there are two possibilities for $x_0$: i) $x_0 \in \text{Int} (\Omega)$; and 2) $x_0 \notin \overline \Omega$, where $\overline \Omega$ represents the closure of set $\Omega$. We will show neither of these can be true. 
		
		i) $x_0 \in \text{Int} (\Omega)$: 
		Since $\phi^f_{-t}(x)$ is a diffeomorphic map, it follows that $\phi^f_{-t}(x_0)\in\text{Int} (\phi^f_{-t}(\Omega))$. Hence, $\phi^f_{-t}(x_0)\notin\partial\phi^f_{-t}(\Omega)$, which is a contradiction.
		%		In terms of Lemma \ref{lem:expansion.1}, we have $\Omega \subset \phi^f_{-t}(\Omega)$. Since $x_0 \in \text{Int}(\Omega)$, there must exist some finite $\tau \;(0< \tau < t)$ such that $\phi^f_{-\tau} (x_0) = x_1 \in \partial \Omega$ (see Fig. \ref{fig.1} (a)). Denote $x'_1=\phi^f_{-t}(x_1)$. According to Property \ref{Prop:1},  we have
		%		$$x'_1=\phi^f_{-t}(x_1)=\phi^f_{-\tau}\circ \phi^f_{-t}(x_0)=\phi^f_{-\tau}(x'_0)$$
		%		
		%		It implies that an arbitrary point on the trajectory between $x_0$ and $x_1$, denoted by $x_2$, will go out of $\phi^f_{-t}(\Omega)$ under the operation of flow mapping $\phi^f_{-t}$. As a consequence, $\phi^f_{-t}(x_2) \notin \phi^f_{-t}(\Omega)$. However, since $x_2 \in \Omega $, Lemma \ref{lem:expansion.1} indicates that $\phi^f_{-t}(x_2) \in \phi^f_{-t}(\Omega)$ holds for any $x\in \Omega$. By contradiction, it is impossible that $x_0 \in \text{Int}(\Omega)$.  
		
		ii) $x_0 \notin \overline \Omega$: 
		In this case, $x_0\in\overline{\Omega}^c$ and $\phi^f_{-t}(x_0)\in\phi^f_{-t}(\overline{\Omega}^c)$. Since $\overline{\Omega}^c$ is open and $\phi^f_{-t}(x)$ is diffeomorphic, we have $\phi^f_{-t}(\overline{\Omega}^c)$ is open. Since $\phi^f_{-t}(x_0)\in\partial (\phi^f_{-t}(\Omega))$, there exists a point $y$ in the neighborhood of $\phi^f_{-t}(x_0)$ such that $y\in\phi^f_{-t}(\overline{\Omega}^c)$ and $y\in\phi^f_{-t}(\Omega)$, which violates the one-to-one property of $\phi^f_{-t}(x)$ since $\overline{\Omega}^c\cap\Omega=\emptyset$.
		%		In this case, since $t>0$, $x_0$ is impossible located outside of $\phi^f_{-t}(\Omega)$ due to the positive invariance of $\phi^f_{-t}(\Omega)$ according to Lemma \ref{lem:expansion.1}. Hence we only need to consider $x_0 \in \phi^f_{-t}(\Omega) \backslash \overline \Omega$. 
		%		
		%		Along the forward time direction, there is $x_0=\phi^f_{t}(x'_0)$, where $x'_0\in \partial \phi^f_{-t}(\Omega)$ and $x_0 \in \phi^f_{-t}(\Omega) \backslash \overline \Omega$. Note that 
		%		\bq \label{eq:boundary.3}
		%		\phi^f_t \circ \phi^f_{-t}(\Omega)=\Omega \subset \overline \Omega
		%		\end{eqnarray}   
		%		Eq.\eqref{eq:boundary.3} indicates that any trajectory starting from point $x\in \phi^f_{-t}(\Omega)$ enters $\overline \Omega$ within the given time $t>0$. 
		%		
		%		Noticing that $\phi^f_{-t}(\Omega) \backslash \overline \Omega$ is an open set, we can always choose a small enough  time $\tau>0$ such that $\phi^f_{\tau}(x_0)=x_1\in \partial \Omega$ (see Fig. \ref{fig.1}b). Let $x'_1=\phi^f_{-t}(x_1)$, then we have 
		%		\begin{eqnarray*}
		%		x'_1=\phi^f_{-t}(x_1)=\phi^f_{\tau}\circ \phi^f_{-\tau} \circ \phi^f_{-t}(x_1)=\phi^f_{\tau}(x'_0)
		%		\end{eqnarray*}
		%		
		%		It indicates that any point on the trajectory between $x'_0$ and $x'_1$, denoted by $x'_2$, cannot enter the set $\overline \Omega$ within the given time $t$ under the flow mapping. This result is in contradiction with Eq.\eqref{eq:boundary.3}. Therefore it is impossible that $x_0 \in \phi^f_{-t}(\Omega) \backslash \overline \Omega$.   
		
		Hence, we conclude that $\partial (\phi _{-t}^{f} (\Omega )) \subseteq \phi _{-t}^{f} (\partial \Omega )$. 
		
		Analogously, we can further prove $\partial (\phi _{-t}^{f} (\Omega ))\supseteq \phi _{-t}^{f} (\partial \Omega )$. Then  $\partial (\phi _{-t}^{f} (\Omega ))=\phi _{-t}^{f} (\partial \Omega )$ is true and 1) is proved. 
		
		$\Rightarrow 2)$. We prove $\forall x\in\partial A(x_s)$ and $\forall \varepsilon>0$, there exists $\tau>0$ such that $\forall t>\tau$, the distance $d(x,\phi _{-t}^{f} (\partial \Omega ))<\varepsilon$.
		
		Note that $ \forall x\in\partial A(x_s), \forall\varepsilon>0$, $\exists x_0\in A(x_s)$ such that $d(x,x_0)<\varepsilon$. Lemma \ref{lem:expansion.1} implies that $\exists\tau>0$ such that $x_0\in\phi _{-t}^{f}(\Omega)$ for all $t>\tau$. By 1) we have $d(x,\phi _{-t}^{f}(\partial\Omega))=d(x,\partial\phi _{-t}^{f}(\Omega))$. So $d(x,\phi _{-t}^{f} (\partial \Omega ))=d(x,\partial\phi _{-t}^{f} (\Omega ))\leq d(x,x_0)<\varepsilon$, which completes the proof.
	\end{proof}
	\subsection{Proof of Theorem \ref{thm:closestUEP}}
	\begin{proof}
		$\Rightarrow 1)$: Since $x^{cl}_u$ is an equilibrium point, the trajectory starting from it is just itself, i.e., $\phi^f_t(x^{cl}_u)=x^{cl}_u$ for any $t\in \rr$. Hence the point $x^{cl}_u$ is invariant under the flow mapping. Then the conclusion is trivial. 
		
		$\Rightarrow$ 2) and 3): Since the level value is chosen as the value of the closest UEP, namely $V(x^{cl}_{u})$, $\partial S_{l}$ is the level set passing through the closest UEP. It holds that $S_{l} \subset A(x_s)$ and $\partial S_{l} \cap \partial A(x_s)= x^{cl}_{u}$. Therefore, Lemma \ref{lem:expansion.1} and Lemma \ref{le:boundary.2} still can apply, directly resulting in both statements 2) and 3). 
	\end{proof}
	\subsection{Proof of Lemma \ref{thm:UEP}}
	\begin{proof}
		Since both $S_{l}$ and $A(x_{s} )$ are positively invariant sets \eqref{eq:system.1}, one can easily show that $D_{l}=S_{l}\cap A(x_s)$ is also a positively invariant set. For any $x\in\ D_l$, it holds that $\phi_t^f(x)\in D_l$ for all $t\in(0,+\infty)$ and $\lim_{t\to +\infty}\phi_t^f(x)=x_s\in D_l$, which implies $x$ is path-connected to $x_s$. Thus, $D_l$ is a connected set.
		Furthermore, by the definition of $D_{l}$, it holds that $x_s\in\text{Int}(D_l)\subseteq D_{l} \subset A(x_s)$. Hence, invoking Lemma \ref{lem:expansion.1}, statements 1), 3) and 4) can be proved following the same arguments as in the proof of Theorem \ref{thm:closestUEP}. 
		
		Finally, statement 2) directly follows from the fact that $\phi^f_{-t}(A(x_s))=A(x_s)$ for all $t\in(0,+\infty)$, which completes the proof. 
	\end{proof}
	
	\subsection{Proof of Theorem \ref{thm:boundary.3}}
	\begin{proof}
		Define the subset $D_{l}:=S_{l} \cap A(x_{s})$. Lemma \ref{thm:UEP} indicates that 
		\begin{eqnarray}
		\label{eq:b.1} 	
		\lim_{t\to +\infty }d(x,\phi _{-t}^{f} (\partial D_{l}))=0, \forall x\in\partial A(x_{s})
		\end{eqnarray}	
		
		Scenarios are trivial when $x^b_u$ has the lowest value, i.e. $x^b_u$ is the closest UEP, or highest value on the boundary. The first scenario is discussed in Theorem \ref{thm:closestUEP}. The second scenario yields $D_l=A(x_s)$ and $\partial S_{l} \cap \overline{A(x_{s})}=\partial A(x_{s} )\cap S_{l}^{c}=x^b_u$, making the statement trivial.
		
		Now we consider nontrivial scenarios when $x^b_u$ has neither the lowest nor highest value on the boundary. We prove in three steps.
		
		\textit{Step 1: boundary partition.}
		
		Recalling Theorem \ref{thm:boundary.1}, it holds that
		\begin{eqnarray}
		\label{eq:b.2}
		\partial A(x_s)= \bigcup _{x^i_{u} \in E\cap \partial A(x_{s} )} \mW^{s} (x^i_{u} ). 
		\end{eqnarray}
		
		Depending on whether they locate within $S_l$ or not, we partition the UEPs on the stability boundary, $x^i_{u} \in E\cap \partial A(x_{s} )$, into two sets 
		\begin{eqnarray*}
			X_{1} &:=&\{ x^i_{u} \in E\cap \partial A(x_{s} )|\;x^i_{u} \in \partial A(x_{s} )\cap S_{l} \}  \\
			X_{2} &:=&\{ x^i_{u} \in E\cap \partial A(x_{s} )|\;x^i_{u} \in \partial A(x_{s} )\cap S_{l}^{c} \}.  
		\end{eqnarray*}
		Clearly, $X_1$ and $X_2$ are both nonempty for nontrivial $x^b_u$. %empty if and only if $\partial A(x_{s} )\cap S_{l}$ ($\partial A(x_{s} )\cap S_{l}^{c}$) is empty. When $X_1=\emptyset$, we have $x_u^b$ is the closest UEP and $S_l\subset A(x_s)$ which is explained in Theorem 6. Clearly we have $X_2\neq\emptyset$ since at least $x_u^b\in X_2$. %However, we may have $\bigcup _{x^i_{u} \in \partial A(x_{s} )\cap S_{l}^{c}} \mW^{s} (x^i_{u})=\emptyset$ when $x_u^b$ has the largest energy value on the boundary. \peng{should we exclude these tow cases in the description of the theorem?}
		
		Correspondingly, the stability boundary, $\partial A(x_{s} )$, can also be divided into two sets:
		\begin{eqnarray*}
			\partial A_{1} :=\bigcup _{x^i_{u} \in X_{1} } \mW^{s} (x^i_{u} ); \quad \quad
			\partial A_{2} :=  \bigcup_{x^i_{u} \in X_{2} } \mW^{s} (x^i_{u} ), 
		\end{eqnarray*}	
		satisfying 
		\begin{eqnarray}
		\label{eq:b.3}
		\left\{
		\begin{array} {lll}
		\partial A_1\cup \partial A_2 &=& \partial A(x_s)\\ 
		\partial A_1\cap \partial A_2&=&\emptyset.
		\end{array} \right.
		\end{eqnarray}	
		Moreover, both $ \partial A_1$ and $\partial A_2$ are invariant sets of \eqref{eq:system.1}.

		Note that the boundary of $D_{l}$ is composed of the boundary of $S_{l}$ inside $\overline {A(x_s)}$ and the boundary of $A(x_s)$ inside $S_{l}$. Therefore, we can also divide $\partial D_{l} $ into two parts, which reads
		\begin{eqnarray*}
			\partial D_{1} := \partial A(x_s) \cap S_{l};\quad \quad  \partial D_{2} :=\partial S_{l} \cap \overline{A(x_{s} )}. 
		\end{eqnarray*}
		Obviously, it holds that
		\begin{eqnarray*}
			%\label{eq:b.4}
			\left\{
			\begin{array} {lll}
				\partial D_1\cup \partial D_2 &=& \partial D_{l}\\   \partial D_1\cap \partial D_2  &=& \emptyset
			\end{array} \right.
		\end{eqnarray*}
		According to Property \ref{Prop:1}, the flow mapping $\phi^f_t$ of \eqref{eq:system.1} is diffeomorphic and hence nonsigular for any $t\in \rr$, yielding 
		\begin{eqnarray}
		\label{eq:b.4}
		\left\{
		\begin{array} {lll}
		\phi^f_{-t}(\partial D_1)\cup \phi^f_{-t}(\partial D_2) &=& \phi^f_{-t}(\partial D_{l})\\   
		\phi^f_{-t}(\partial D_1) \cap \phi^f_{-t}(\partial D_2)  &=& \emptyset.
		\end{array} \right.
		\end{eqnarray}
		%Furthermore, in terms of definitions of stable and unstable manifolds, the following equation holds
		\textit{Step 2:
			We claim the following equation holds:
			\begin{eqnarray}
			\label{eq:b.5}
			\lim_{t\to +\infty} \phi^f_{-t}(\partial D_1) = \partial A_1.
			\end{eqnarray} } 
		For all $t\in\rr$, we have
		$$\phi^f_{-t}(\partial D_1)=\phi^f_{-t}(\partial A(x_s) \cap S_l)=\phi^f_{-t}(\partial A(x_s))\cap\phi^f_{-t}(S_l).$$
		Invoking $\phi^f_{-t}(\partial A(x_s))=\partial A(x_s),\;\forall t\in\rr$, we obtain
		$$\phi^f_{-t}(\partial D_1)=\partial A(x_s)\cap\phi^f_{-t}(S_l).$$
		Since $\phi^f_{-t_1}(S_l)\subset\phi^f_{-t_2}(S_l),\;\forall t_1<t_2$, the sequence of sets $\{\phi^f_{-t}(\partial D_1)\}_t$ is monotonically increasing as $t\to+\infty$. 	
		Thus the limit in \eqref{eq:b.5} exists and we have
		$$\lim_{t\to +\infty} \phi^f_{-t}(\partial D_1) =\bigcup_{t>0}\phi _{-t}^{f} (\partial D_1).$$
		By the invariance property of $\partial A_1$, we have 
		$$\lim_{t\to +\infty} \phi^f_{-t}(\partial D_1)\subseteq\partial A_1.$$
		Furthermore, $\forall x\in\partial A_1$ there is a $x_u^i\in X_1$ such that $\lim_{t\to +\infty} \phi^f_{t}(x)=x_u^i$ and $\phi^f_{t}(x)\in\partial A_1, \forall t\in\rr$. Thus $\exists \tau>0$ such that $\phi^f_{\tau}(x)\in\partial D_1$ which means $x\in\phi^f_{-\tau}(\partial D_1)$. By the arbitrary of $x$, we have $\partial A_1\subseteq\lim_{t\to +\infty} \phi^f_{-t}(\partial D_1)$. Therefore, the equality in \eqref{eq:b.5} holds.
		
		\textit{Step 3:
			We claim that $\forall x\in\partial A_2$, $\forall \varepsilon>0$, $\exists \tau>0$ such that $d(x,\phi_{-t}^f(\partial D_2))<\varepsilon$, $\forall t>\tau$.}
		
		For all $ x\in\partial A_2\setminus\overline{\partial A_1}$, the distance $d(x,\overline{\partial A_1})$ is strictly bounded away from zero. By \eqref{eq:b.1} and \eqref{eq:b.5} we have
		\begin{equation}\label{eq:b.6}
		\lim_{t\to +\infty}d(x,\phi _{-t}^{f} (\partial D_2))=0.
		\end{equation}
		
		Consider $x\in\overline{\partial A_1}\cap\partial A_2$. Note that $\partial A_2$ is not isolated as $x^b_u$ is not the highest energy point on the boundary. So for any $\varepsilon>0$, there exists a point $x'\in\partial A_2\setminus\overline{\partial A_1}$ such that $d(x,x')<{\varepsilon}/{2}$. Thus, there exists $\tau>0$ such that for all $t>\tau$
		$$d(x',\phi _{-t}^{f} (\partial D_2))<\frac{\varepsilon}{2}.$$
		It follows from the triangle inequality that
		$$d(x,\phi _{-t}^{f} (\partial D_2))\leq d(x,x')+d(x',\phi _{-t}^{f} (\partial D_2))<\varepsilon.$$
		Hence, we have
		$$\lim_{t\to +\infty}d(x,\phi _{-t}^{f} (\partial D_2))=0,\qquad\forall x\in\partial A_2$$
		
		According to the definitions of $\partial D_2$ and $\partial A_2$, it directly results in 
		\begin{eqnarray*}
			\mathop{\lim}\limits_{t\to +\infty }d(x, \phi _{-t}^{f} (\partial S_{l} \cap \overline{A(x_{s})}))=0, \forall x\in\bigcup _{x^i_{u} \in \partial A(x_{s} )\cap S_{l}^{c}} \mW^{s} (x^i_{u}) 
		\end{eqnarray*}
		
		The proof is completed.                                              
	\end{proof}
	%	All the EPs in this area are listed in Tab. \ref{tab:eps}, including one SEP and twelve UEPs on the stability boundary. 
	%	
	%	\begin{table}[htb]
	%		\renewcommand{\arraystretch}{1.3}
	%		\centering
	%		\caption{\\ \textsc { Equilibrium Points of the System}}
	%		\label{tab:eps}
	%		\begin{tabular}{c c c}
	%			\hline
	%			\textbf{Equilibrium Point} & \textbf{Type} & \textbf{Energy Value} \\ \hline 
	%			(0.02801,0.06403) & 0(SEP) & 0 \\ \hline 
	%			(0.4667,3.1149) & 1 & 3.7770 \\ \hline 
	%			(-3.037,0.3341) & 1 & 6.0490 \\ \hline 
	%			(3.2458,0.3341) & 1 & 5.9234 \\ \hline 
	%			(3.0407,3.2232) & 1 & 5.6235 \\ \hline 
	%			(-3.2425,-3.060) & 1 & 6.3775 \\ \hline 
	%			(0.0467,-3.1683) & 1 & 4.3186 \\ \hline 
	%			(2.608,4.225) & 2 & 5.7654 \\ \hline 
	%			(3.597,1.575) & 2 & 6.0105 \\ \hline 
	%			(2.608,-2.028) & 2 & 6.3939 \\ \hline 
	%			(-2.686,-4.708) & 2 & 6.7645 \\ \hline 
	%			(-3.675,-2.028) & 2 & 6.5196 \\ \hline 
	%			(-2.686,1.575) & 2 & 6.1362 \\ \hline 
	%		\end{tabular}
	%	\end{table}
	%	
	%	Type 1 UEPs and type 2 UEPs in Tab. \ref{tab:eps}, along with the exact stability boundary, are depicted in Fig. \ref{fig:EPs}.
	%	
	%	\begin{figure}[htp]
	%		\centering
	%		\includegraphics[width=0.8\columnwidth]{image_EPs}
	%		\caption{The EPs and exact stability boundary of the system.}
	%		\label{fig:EPs}
	%	\end{figure}
	%	

\end{document}